\documentclass[manuscript,nonacm]{acmart}
\setcopyright{none}
\usepackage[T1]{fontenc} 
\usepackage{appendix}
\usepackage{graphicx}
\usepackage[colorinlistoftodos,prependcaption,textsize=tiny]{todonotes}
\usepackage[ruled,linesnumbered,algonl,procnumbered]{algorithm2e}
\usepackage{amsmath}
\SetInd{0.4em}{0.8em} 
\SetAlFnt{\footnotesize}
\SetKwProg{Fn}{Function}{}{}%
\SetKwProg{Macro}{Macro}{}{}%
\SetKwProg{struct}{Struct}{}{}%
\usepackage{wrapfig}
\usepackage{multicol}
\usepackage{amsthm}
\usepackage{algpseudocode}
\usepackage[bottom]{footmisc}
\usepackage{xstring}
\usepackage{enumitem}
\usepackage{subcaption}
\usepackage{listings}
\usepackage{float}
\usepackage{appendix}

\usepackage{mathtools}
\usepackage{amsfonts}
\usepackage{stmaryrd}
\usepackage{xcolor}
\usepackage{tikz}
\usepackage{url,xspace}
\newcommand{\figref}[1]{Figure \ref{#1}}

\newcommand{\relaxation}{relaxation\xspace}
\newcommand{\relaxing}{relaxing\xspace}

\newcommand{\relaxed}{relaxed\xspace}
\newcommand{\Relaxation}{Relaxation\xspace}

\newcommand{\accuracy}{\textit{accuracy}\xspace}

\newcommand{\access}{access\xspace}
\newcommand{\accessp}{access-point\xspace}

\newcommand{\substructure}{\textit{sub-structure}\xspace}

\newcommand{\substructures}{\textit{sub-structures}\xspace}
\newcommand{\koo}{\textit{k-out-of-order}\xspace}

\newcommand{\wDDc}{\textit{2Dc}\xspace}
\newcommand{\wDDd}{\textit{2Dd}\xspace}
\newcommand{\DDc}{\textit{2Dc-structure}\xspace}
\newcommand{\DDd}{\textit{2Dd-structure}\xspace}

\newcommand{\processor}{thread\xspace}
\newcommand{\processors}{threads\xspace}

\newcommand{\Processors}{Threads\xspace}
\newcommand{\hopo}{hop\xspace}
\newcommand{\hopping}{hopping\xspace}
\newcommand{\hopped}{hopped\xspace}
\newcommand{\hopos}{\textit{hops}\xspace}
\newcommand{\Hopos}{\textit{Hops}\xspace}
\newcommand{\window}{\textit{Window}\xspace}
\newcommand{\wincoupled}{\textit{WinCoupled}\xspace}
\newcommand{\windecoupled}{\textit{WinDecoupled}\xspace}
\newcommand{\windows}{\textit{Windows}\xspace}
\newcommand{\Window}{\textit{Window}\xspace}

\newcommand{\winmax}{\textit{$Win_{max}$}\xspace}
\newcommand{\winmin}{\textit{$Win_{min}$}\xspace}
\newcommand{\globalcounter}{\textit{$Win_{max}$}\xspace}
\newcommand{\winlimit}{\textit{$[\winmax,\winmin]$}\xspace}
\newcommand{\shiftup}{\textit{$shift_{up}$}\xspace}
\newcommand{\shiftdown}{\textit{$shift_{down}$}\xspace}

\newcommand{\depth}{\textit{depth}\xspace}
\newcommand{\width}{\textit{width}\xspace}

\newcommand{\slideop}{\textit{$W_{shift}$}\xspace}
\newcommand{\localcounter}{\textit{sub-count}\xspace}
\newcommand{\localgetcounter}{\textit{Get-sub-count}\xspace}
\newcommand{\localputcounter}{\textit{Put-sub-count}\xspace}

\newcommand{\putop}{\textit{Put}\xspace}
\newcommand{\getop}{\textit{Get}\xspace}

\newcommand{\random}{\textit{Random}\xspace}
\newcommand{\roundrobin}{\textit{Round-Robin}\xspace}
\newcommand{\randomc}{\textit{Random-C2}\xspace}

\newcommand{\sDD}{\textit{2D-Stack}\xspace}
\newcommand{\sDDd}{\textit{2Dd-Stack}\xspace}
\newcommand{\sDDc}{\textit{2Dc-Stack}\xspace}
\newcommand{\srandom}{\textit{S-random}\xspace}
\newcommand{\srandomc}{\textit{S-random-c2}\xspace}
\newcommand{\srobin}{\textit{S-robin}\xspace}
\newcommand{\kstack}{\textit{k-Stack}\xspace}
\newcommand{\elimination}{\textit{Elimination}\xspace}

\newcommand{\tsstack}{\textit{TS-Stack}\xspace}
\newcommand{\substak}{\textit{sub-stack}\xspace}
\newcommand{\substaks}{\textit{sub-stacks}\xspace}
\newcommand{\popo}{\textit{Pop}\xspace}
\newcommand{\pusho}{\textit{Push}\xspace}
\newcommand{\descriptor}{\textit{descriptor}\xspace}

\newcommand{\qDD}{\textit{2D-Queue}\xspace}
\newcommand{\qDDd}{\textit{2Dd-Queue}\xspace}
\newcommand{\qsegment}{\textit{Q-segment}\xspace}
\newcommand{\qrandom}{\textit{Q-random}\xspace}
\newcommand{\qrandomc}{\textit{Q-random-c2}\xspace}
\newcommand{\qrobin}{\textit{Q-robin}\xspace}
\newcommand{\msqueue}{\textit{MS-queue}\xspace}
\newcommand{\wfqueue}{\textit{wfqueue}\xspace}
\newcommand{\lru}{\textit{lru}\xspace}
\newcommand{\subq}{\textit{sub-queue}\xspace}
\newcommand{\subqs}{\textit{sub-queues}\xspace}
\newcommand{\globalcounterEnq}{\textit{$Win_{max}^{put}$}\xspace}
\newcommand{\globalcounterDeq}{\textit{$Win_{max}^{get}$}\xspace}

\newcommand{\subqueue}{\textit{sub-queue}\xspace}
\newcommand{\subqueues}{\textit{sub-queues}\xspace}
\newcommand{\enqop}{\textit{Enqueue}\xspace}
\newcommand{\deqop}{\textit{Dequeue}\xspace}

\newcommand{\cDD}{\textit{2D-Counter}\xspace}
\newcommand{\cDDd}{\textit{2Dd-Counter}\xspace}
\newcommand{\cDDc}{\textit{2Dc-Counter}\xspace}
\newcommand{\crandom}{\textit{C-random}\xspace}
\newcommand{\crandomc}{\textit{C-random-c2}\xspace}
\newcommand{\crobin}{\textit{C-robin}\xspace}
\newcommand{\cfaa}{\textit{C-FAA}\xspace}
\newcommand{\increment}{\textit{increment}\xspace}
\newcommand{\decrement}{\textit{decrement}\xspace}
\newcommand{\subcounter}{\textit{sub-counter}\xspace}
\newcommand{\subcounters}{\textit{sub-counters}\xspace}

\newcommand{\dDDd}{\textit{2Dd-Deque}\xspace}
\newcommand{\dDD}{\textit{2D-Deque}\xspace}
\newcommand{\subdeque}{\textit{sub-deque}\xspace}
\newcommand{\subdeques}{\textit{sub-deques}\xspace}
\newcommand{\drobin}{\textit{Deque-robin}\xspace}
\newcommand{\drandom}{\textit{Deque-Random}\xspace}
\newcommand{\dmaged}{\textit{Deque-Maged}\xspace}
\newcommand{\dsundell}{\textit{Deque-Sundell}\xspace}
\newcommand{\pushleftop}{\textit{PushLeft}\xspace}
\newcommand{\popleftop}{\textit{PopLeft}\xspace}
\newcommand{\pushrightop}{\textit{PushRight}\xspace}
\newcommand{\poprightop}{\textit{PopRight}\xspace}

\newcommand{\codetxt}[1]{\texttt{#1}\xspace}

\newcommand{\caeo}{\codetxt{CAS}}
\newcommand{\faa}{\codetxt{FAA}}

\newcommand{\multisocket}{\textit{Multi-S}\xspace}
\newcommand{\singlesocket}{\textit{Single-S}\xspace}

\newcommand{\sjwitems}{\textit{$N^{active}_i$}\xspace}

\newcommand{\sjitems}{\textit{$N_i$}\xspace}
\newcommand{\witems}{\textit{$N^{active}$}\xspace}
\newcommand{\fshift}{\textit{$shift^{up}$}\xspace}
\newcommand{\bshift}{\textit{$shift_{down}$}\xspace}
\newcommand{\wsize}{\textit{$K$}\xspace}

\newcommand{\shifted}{shifted\xspace}
\newcommand{\shift}{shift\xspace}
\newcommand{\shifts}{shifts\xspace}
\newcommand{\shifting}{shifting\xspace}
\newcommand{\Shifting}{Shifting\xspace}
\newcommand{\sta}[1]{\ema{\mathcal{S}_{#1}}}
\newcommand\floor[1]{\lfloor#1\rfloor}

\newtheorem{theorem}{\textbf{Theorem}}
\newtheorem{lemma}[theorem]{\textbf{Lemma}}

\newcommand{\treiberop}{\textit{$op$}\xspace}

\newcommand{\ignore}[1]{}

\newcommand{\ie}{\textit{i.e.}\xspace}

\newcommand{\nil}{\textit{NULL}\xspace}
\newcommand{\ema}[1]{\ensuremath{#1}\xspace}

\newcommand{\pro}[1]{\ema{\mathbb{P}\left(#1\right)}}
\newcommand{\expe}[1]{\ema{\mathbb{E}\left(#1\right)}}

\newcommand{\intedef}[2][0]{\ema{\left\llbracket #1,#2 \right\rrbracket}}
\newcommand{\inte}[2]{\ema{\left\llbracket #1,#2 \right\rrbracket}}

\begin{document}

\title[Relaxing Concurrent Data-structure Semantics for Increasing Performance]{Relaxing Concurrent Data-structure Semantics for Increasing Performance: A Multi-structure 2D Design Framework}.

\author{Adones Rukundo}
\email{adones@chalmers.se, adones@must.ac.ug}
\author{Aras Atalar}
\email{aaras@chalmers.se}
\author{Philippas Tsigas}
\email{tsigas@chalmers.se}
\affiliation{%
  \institution{Chalmers University of Technology}
  \city{Gothenburg}
  \country{Sweden}
}

\renewcommand{\shortauthors}{Adones Rukundo, et al.}

\begin{abstract}
There has been a significant amount of work  in the literature proposing semantic relaxation of concurrent data structures for improving scalability and performance. By relaxing the semantics of a data structure, a bigger design space, that allows weaker synchronization and more useful parallelism, is unveiled. Investigating new data structure designs, capable of trading semantics for achieving better performance in a monotonic way, is a major challenge in the area. We algorithmically address this challenge in this paper.
 
To address this challenge, we present an efficient lock-free, concurrent data structure design framework for {\it out-of-order} semantic relaxation. Our framework introduces a new two dimensional algorithmic design, that uses multiple instances of an implementation of the given data structure. The first dimension of our design is the number of data structure instances onto which operations are spread to, in order to achieve increased parallelism through disjoint memory accesses. The second dimension is the number of consecutive operations of a single \processor that can stay at the same data structure instance in order to benefit from data locality. Our design can flexibly explore this two-dimensional space to achieve the property of monotonically increasing throughput performance via relaxing concurrent data structure semantics within a tight deterministic relaxation bound, as we prove in the paper.
 
We show how our framework can instantiate lock-free {\it out-of-order} queues, stacks, counters and dequeues. We provide implementations of these \relaxed data structures and evaluate their performance and behaviour on two parallel architectures. The experimental evaluation shows that our two-dimensional data structures: i)  significantly outperform the respected previous proposed ones with respect to scalability and throughput performance and ii)  monotonically increase throughput as relaxation  increases. 

\end{abstract}
\maketitle

\section{Introduction}
\label{sec:introduction}

\begin{figure}[b]
    \begin{minipage}[c]{0.21\textwidth}
        \includegraphics[scale=0.17]{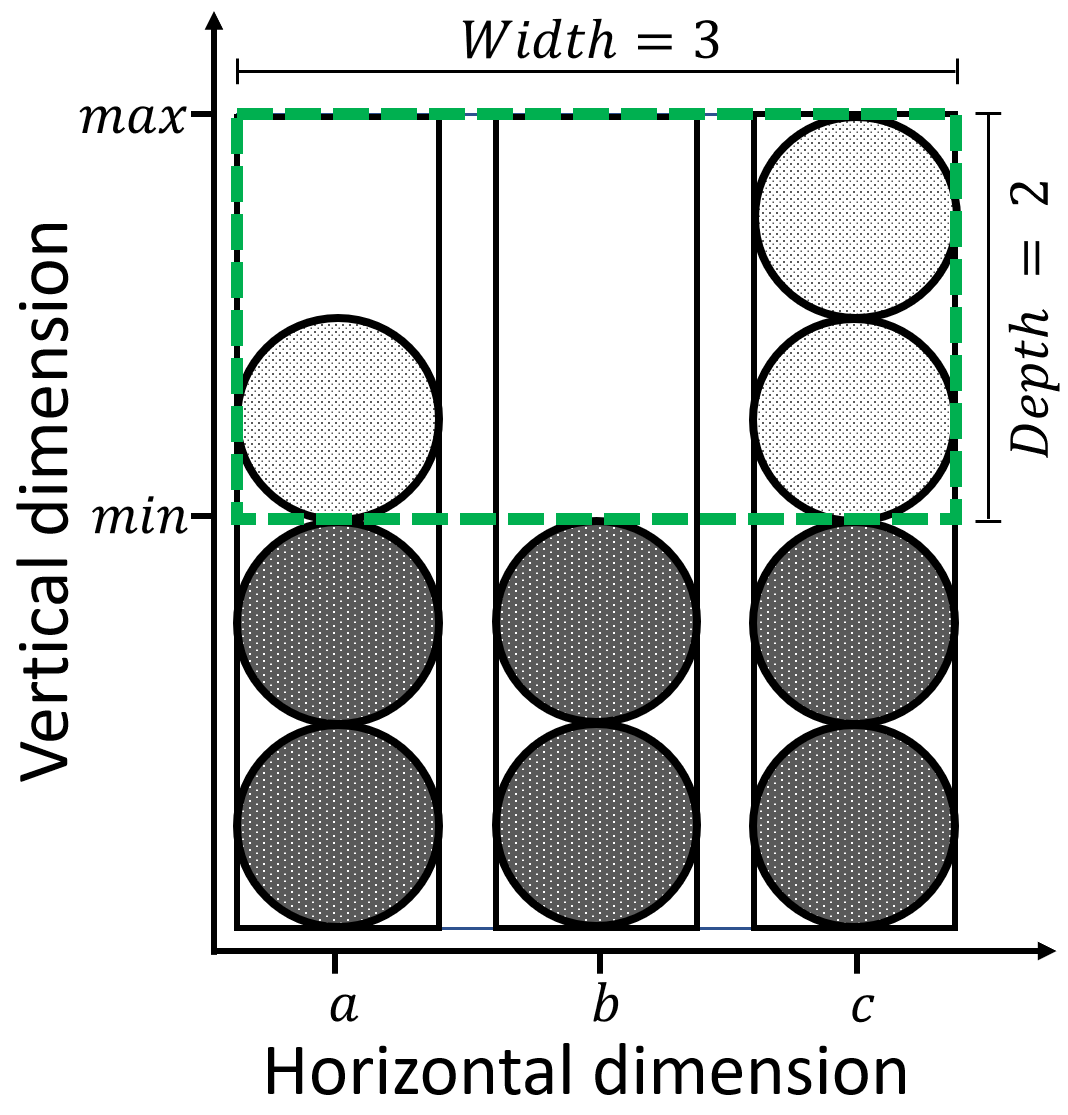}
    \end{minipage}\hfill
    \begin{minipage}[c]{0.79\textwidth}
        \caption{An illustration of our 2D design using a Stack as an example. The example employs three \substaks $a$, $b$ and $c$. The relaxation bound $k$ is proportional to the area of the green dashed rectangle in which stack operations are bounded to occur. Circles represent stack items. Grey circles depict items that can be Poped. $a$ can be used for both \pusho and \popo. $b$ can be used for \pusho but not for \popo. $c$ can be used for \popo but not for \pusho.}
        \label{fig:2Dframework}
    \end{minipage}\hfill
\end{figure}

Concurrent data structures allow operations by multiple \processors to concurrently access the data structures, and this requires synchronisation to guarantee consistency with respect to their sequential semantics \cite{dijkstra1965solution,Dijkstra:1968:SLS:363095.363143}. The synchronisation of concurrent accesses aims at enforcing some notion of atomicity, where, an operation appears to occur at a single instant between its invocation and its response. A concurrent data structure is typically designed around one or more synchronisation \access points, from where \processors compute, consistently, the current state of the data structure. Synchronisation is vital in achieving consistency and cannot be eliminated \cite{Attiya:2011:LOE:1925844.1926442}. Whereas this is true, synchronisation might generate contention in memory resources hurting scalability and performance.

The necessity of reducing contention at the synchronisation \access points, and consequently improving scalability, is and has been a major focus for concurrent data structure researchers. Aiming to address this challenge, techniques like: elimination \cite{afek2010scalable,hendler2010scalable,shavit1995elimination}, combining \cite{shavit2000combining}, dynamic elimination-combining \cite{DBLP:journals/corr/abs-1106-6304}, back-off, operation buffers and batching have been proposed. To address, in a more significant way, the challenge of scalability bottlenecks of concurrent data structures, it has been proposed that the semantic legal behaviour of data structures should be extended \cite{Shavit:2011:DSM:1897852.1897873}. This line of research has led to the introduction of an extended set of weak semantics including weak internal ordering, weakening consistency and semantic \relaxation. 

One of the main definitions of data structure semantic \relaxation, proposed and used in the literature, is \koo \cite{afek2010quasi,henzinger2013quantitative,Talmage2017,Haas:2013:DQS:2482767.2482789,rihani2015brief,wimmer2015lock}. \koo semantics allow operations to occur out of order within a given $k$ bound, e.g. a pop operation of a \koo stack can remove any item among the $k$ topmost stack items. By allowing a \popo operation to remove any item among the $k$ topmost stack items, the extended stack semantics do not anymore impose by definition a single \access point for all its operations. Thus, allowing for potentially  more efficient stack designs with lower synchronisation overhead.

Semantics \relaxation can be exploited to better use disjoint \access points to achieve improved parallelism and \processor local data processing \cite{DBLP:conf/wdag/RukundoAT19}. Disjoint \access is popularly achieved by distributing data structure operations over multiple disjoint instances of a given data structure \cite{Haas:2013:DQS:2482767.2482789,rihani2015brief,DBLP:conf/concur/HaasHHKLPSSV16,DBLP:conf/esa/Williams0D21}. On the other hand, the locality is generally achieved through letting a  \processor stick to the same memory location for a given number of consecutive operations \cite{wimmer2015lock,gidron2012salsa,DBLP:conf/concur/HaasHHKLPSSV16}. 

In this paper, we introduce an efficient two-dimensional algorithmic design framework, that uses disjoint multiple instances (\substructures) of a given data structure as shown in Figure \ref{fig:2Dframework}. The first dimension of the framework is the number of \substructures operations are spread to, in order to benefit from parallelism through disjoint \access points. The second dimension is the number of consecutive operations that can occur on the same \substructure in order to benefit from data locality by a \processor sticking to the same memory location for a given number of operations. We use two parameters to control the dimensions; \width for the first dimension (horizontal) and \depth for the second dimension (vertical) as shown in Figure \ref{fig:2Dframework}.

A \processor can operate on a given \substructure for as long as a set of conditions hold for this \substructure with respect to the operation on hand, in this case we call the \substructure valid. As an example, a valid \substructure might be one that does not exceed \emph{max} or go below \emph{min}, as depicted by the dashed green rectangle in Figure \ref{fig:2Dframework}. Validity conditions make \substructures valid or invalid for a given operation. This implies that \processors have to search for a valid \substructure for the respective operation, increasing operation cost (latency). Our framework efficiently controls this search overhead by limiting the number of \substructures utilised, and by allowing a \processor to operate on the same \substructure consecutively for as long as the validity conditions hold. Using these validity conditions, our framework generates tenable and tunable \relaxation behaviour, tightly bound by a deterministic \relaxation bound.

Our algorithmic design framework can be used to extend existing lock-free data structure algorithms to achieve \koo semantics. This can be achieved with minimal modifications to the original data structure algorithm as we later show in Section \ref{sec:algorithms}. Using our framework, we derive semantically \relaxed lock-free \koo stacks, queues, dequeue and counters together with proof of correctness. Experimental evaluation shows that the derived data structures significantly outperform all respective previous known data structure implementations and monotonically increase throughput as relaxation increases.

The rest of the paper is structured as follows. In Section \ref{sec:relatedwork} we discuss the literature related to this work. We present the $2D$ framework in Section \ref{sec:2Dtechnique} followed by the optimisation discussion in Section \ref{sec:optimization} and complexity analysis in Section \ref{sec:complexity}. The derived data structure algorithms are presented in Section \ref{sec:algorithms}. With the aim to enrich the  evaluation of our algorithmic design, a pallet of additional multi structure algorithms are presented in Section \ref{sec:otheralgorithms}. An experimental evaluation is presented and discussed in Section \ref{sec:evaluation}. The paper concludes in Section \ref{sec:conclusion}.
\section{Related Work}
\label{sec:relatedwork}
Recently, data structure semantic \relaxation has attracted the attention of researchers, as a promising direction towards improving concurrent data structures' scalability \cite{henzinger2013quantitative,Talmage2017,Shavit:2011:DSM:1897852.1897873,DBLP:conf/esa/Williams0D21}. It has also been shown that small changes on the semantics of a data structure can have a significant effect on the computational power of the data structure \cite{Shavit:2016:CRD:3011492.3011515}. The interest in semantic \relaxation is further supported by the need to optimise concurrent data structures for the ever growing parallel/multi-core computing hardware. As discussed earlier in Section \ref{sec:introduction}, one of the main definitions of semantic \relaxation proposed and used in the literature is \koo.

Using the \koo definition, a segmentation technique has been proposed in \cite{afek2010quasi}, later revisited in \cite{henzinger2013quantitative} realising a \relaxed Stack (\kstack) and FIFO Queue (\qsegment) both with \koo semantics. The technique involves a linked-list of memory segments with $k$ number of indexes on which an item can be added or removed per index. The stack items are accessed through the topmost segment, whereas the queue has a tail and head segment from which \enqop and \deqop can occur respectively. Segments can be added and removed. \Relaxation is only controlled through varying the number of indexes per segment. As discussed in Section \ref{sec:introduction}, increasing the number of indexes increases operation latency and later becomes a performance bottleneck. This limits the performance benefits of the technique to a small range of \relaxation values. 

Also, load balancing together with multiple queue instances (\subqueues) has been used to design a \relaxed FIFO queue (\lru) with \koo semantics \cite{Haas:2013:DQS:2482767.2482789}. Each \subqueue maintains two counters, one for \enqop another for \deqop, while two global counters, one for \enqop and one for \deqop maintain the total number of operations for all \subqueues. The global counters are used to calculate the expected number of operations on the least-recently-used \subqueue. \Processors can only operate on the least-recently-used \subqueue. This implies that for every operation, \processors must synchronise on the global counter, making it a sequential bottleneck. Moreover, \processors have to search for the least-recently-used \subqueue leading to high operation and memory latency due to lack of \processor locality.

Random distribution of data structure operations over multiple instances (\substructures) of a given data structure has been used to design \relaxed priority queues \cite{rihani2015brief}, and has been tested on other data structures including stacks, counters and FIFO queues \cite{DBLP:conf/wdag/RukundoAT19}. For each operation, a \processor randomly selects a \substructure onto which it goes ahead to perform that given operation. To improve on the quality (order guarantees), a given \processor can select a set of \substructure uniformly at random, and then selects an appropriate \substructure (according to the data structure semantics) from the given set of \substructures onto which the \processor proceeds to operate. The technique significantly reduces contention where the number of \substructures is higher than the number of \processors and operations are distributed uniformly at random. However, as discussed in \cite{DBLP:conf/spaa/RukundoAT22}, as the number of \substructures increase, memory latency increases due to lack of \processor locality and limits the performance benefits gained from reduced contention. In terms of \koo, the technique does not provide for a deterministic \relaxation bound ($k$) but rather a probabilistic one under specific assumptions \cite{Alistarh:2017:PCP:3087801.3087810}.

The proposed \relaxation techniques, mentioned above, apply \relaxation in one dimension, i.e, increase the number of disjoint \access points to improve parallelism and reduce contention. However, this also increases operation latency due to the increased number of \access points to select from, and memory latency due to lack of locality. Without a remedy to this downside, the proposed techniques cannot provide monotonic \relaxation for better performance. Towards addressing the challenge of locality, some designs such as that of a priority \cite{rihani2015brief} have been revisited \cite{DBLP:conf/esa/Williams0D21} to introduce \processor locality. The authors noted that a \processor accesses several cache lines from randomly chosen \substructures which are rarely reused but usually cause cache invalidation costs later, a downside also studied and discussed in \cite{DBLP:conf/spaa/RukundoAT22}. To improve locality \cite{DBLP:conf/esa/Williams0D21}, three measures are applied; insertion and deletion buffers, operation batching and a \processor sticking on the same set of \substructures for several consecutive operations. Operation buffers and batching have been frequently used in the literature to reduce contention at a specific memory location \cite{10.1145/3210377.3210388,alistarh2015spraylist,10.1007/978-3-030-85665-6_24,wimmer2015lock}. Buffering and batching can also help to improve locality where \processors can have exclusive buffers (similar to sticking to a given memory location) \cite{wimmer2015lock}, otherwise, the technique would also suffer from memory latency due to cache invalidation costs. One of the most efficient techniques for improving \processor locality is having a \processor operate on the same memory location (sticking) for several consecutive operations to avoid cache line sharing and invalidation \cite{DBLP:conf/spaa/RukundoAT22}. 

Although \processor memory sticking can improve locality, it can significantly reduce the quality (order guarantees) of a relaxed data structure, outweighing the performance gained from locality as observed in \cite{DBLP:conf/esa/Williams0D21}. As the \processor stickiness increases, the expected quality drops at a higher rate than the measured increase in throughput performance. It is also noted that \processor stickiness under random distribution complicates the provision of order guarantees. In other works, the \processor buffer technique has been used to realise a relaxed priority queue \cite{wimmer2015lock}. Here each \processor is assigned a local buffer (a sticky \substructure) onto which it performs its priority queue operations, later when the buffer is full, the given \processor flushes its buffer to a shared priority queue. On the other hand when the \processor buffer is empty, the given \processor can get (work stealing) items from other available \processor buffers or get items from the shared priority queue. In this case, \processor buffers provides a deterministic \relaxation bound, however, the bound is a product of the degree of stickiness (number of operations per buffer) and number of \processors. This implies that an increase in number of \processors or buffer size leads to a drop in data quality, facing the same challenge as discussed above. Buffer techniques can face a challenge of merging data especially for priority queues where keys might have to be inserted independently at different points within the data structure as observed in \cite{wimmer2015lock}. As noted in \cite{wimmer2015lock} work stealing can increase the complexity of the data structure design and also reduce cache efficiency since it involves spying and modifying other \processor buffers. 

Without using \substructures, a relaxed skip list based priority queue is proposed, where the delete operations are randomly distributed (sprayed) to different points within a skip list \cite{alistarh2015spraylist}. The delete operations avoid the sequential bottleneck at the head by randomly walking the list and returning any element among the first $O(plog^3p)$, where $p$ is the number of threads. The technique reduces contention that can arise from concurrent delete operations trying to access and return the same element at the head. The technique provides a probabilistic \relaxation bound that is proportional to the number of \processors. This implies, that trying to increase throughput through increasing number of \processors reduces the quality of the data. 

Apart from semantic \relaxation, other design strategies such as internal weak ordering \cite{Dodds:2015:SCT:2775051.2676963}, and local linearizability \cite{DBLP:conf/concur/HaasHHKLPSSV16} have been proposed to extend semantics for improving scalability. Weak internal ordering has been proposed and used to implement a timestamped stack (\tsstack) \cite{Dodds:2015:SCT:2775051.2676963}, where \pusho timestamps each pushed item. Each \processor has its local buffer onto which it performs \pusho operations. However, \popo operations pay the cost of searching for the latest item. In the worst case, \popo operations might contend on the same latest item if there are no concurrent \pusho operations. This leads to search retries, especially for workloads with higher \popo rates (bursts) than \pusho ones. 

Local linearizability has also been proposed for designing efficient concurrent data structures such as; FIFO queues and Stacks \cite{DBLP:conf/concur/HaasHHKLPSSV16}. Locally linearizable FIFO queues and Stacks proposed in \cite{DBLP:conf/concur/HaasHHKLPSSV16} relies on multiple instances of a given data structure. Each \processor is assigned an instance on which it locally linearizes all its operations. Operations: \enqop (FIFO queue) or \pusho (Stack) occur on the assigned instance for a given \processor, whereas, \deqop or \popo can occur on any of the available instances. With \deqop or \popo occurring more frequently, contention quickly builds as \processors try to access remote buffers. The \processors also lose the locality advantage while accessing remote buffers, cancelling out the caching advantage especially for single access data structures such as the Stack \cite{Hackenberg:2009:CCA:1669112.1669165,David:2013:EYA:2517349.2522714,Schweizer:2015:ECA:2923305.2923811}. 

Without \relaxing or weakening consistence of the data structure sequential semantics, elimination \cite{afek2010scalable,shavit1995elimination,DBLP:journals/corr/abs-1106-6304} and combining \cite{shavit2000combining} have also been proposed as way to improve scalability by reducing contention at specific points within a given data structure. Elimination implements a collision path on which different concurrent operations (such as \pusho and \popo for the case of a stack) try to collide and cancel out, without accessing a joint \access point (such as a stack's top), otherwise, they proceed to access the central structure \cite{hendler2010scalable,Moir:2005:UEI:1073970.1074013}. Such operation pairs create disjoint collisions that can be executed in parallel.
Combining, on the other hand, allows operations from multiple \processors to be combined and executed by a single \processor without the other \processors contending on the central structure \cite{Fatourou:2012:RCS:2145816.2145849,Hendler:2010:FCS:1810479.1810540}. However, their performance depends on the specific workload characteristics. Elimination mostly benefits symmetric workloads, whereas combining mostly benefits asymmetric workloads. Furthermore, the central structure sequential bottleneck problem still persists. 
\section{The 2D Framework}
\label{sec:2Dtechnique}

In this section, we describe our disjoint multi-structure $2D$ design framework. The 2D framework uses disjoint multiple instances (\substructures) of the given data structure as depicted in \figref{fig:2Dframework}. \Processors can select and operate on any of the \substructures following a maximum and minimum operation count threshold. Herein, \emph{operation} refers to the process that updates the data structure state by adding (\putop) or removing (\getop) an item (\pusho and \popo respectively for the stack example). Each \substructure holds a counter (\localcounter) that counts the number of local successful operations. 

A combination of operation count threshold and number of \substructures, form a logical count period, we refer to it as \window, depicted by the dashed green rectangle in Figure \ref{fig:2Dframework}. The \window limits the number of operations that can occur on each \substructure (\winlimit): maximum (\winmax) and minimum (\winmin) operation count threshold for all \substructures, for a given period. This implies, that for a given period, a \substructure can be operated on (valid) or not (invalid) as exemplified in \figref{fig:2Dframework}, and a \window can be full or empty. The \window is \emph{full} if all \substructures have the maximum number of operations ($\localcounter=\winmax$), \emph{empty}, if all \substructures have the  minimum number of operations ($\localcounter=\winmin$). The \window is defined by two parameters; \width and \depth as exemplified in Figure \ref{fig:2Dframework}. $\width=\#\substructures$, and $\depth=\winmax-\winmin$.

To validate a \substructure, its \localcounter is compared with either \winmax or \winmin depending on the operation; $\localcounter<\winmax$ or $\localcounter>\winmin$. If the given \substructure is invalid, the \processor has to \hopo to another \substructure until a valid \substructure is found (validity is operation specific as we discuss later). If a \processor cannot find a valid \substructure, then, the \window is either full or empty. The \processor will then, either increment or decrement both \winmax and \winmin, the process we refer to as, \window \shifting. A \window can \shift up or down, and is controlled by \shiftup or \shiftdown values respectively, where, $0<\shiftup,\shiftdown\leq\depth$. The \processor can only shift the \window by a given \shift value. \shiftup and \shiftdown values can be configured differently to optimise for different workloads and execution environment.

\Relaxation is controlled by the two parameters that define a \window: \width and \depth. \width and \depth provide a deterministic \relaxation bound for the derived \koo \relaxed data structure as proved in Section \ref{sec:modelcorrectness}. We define and present the implementation of two types of \windows: \wincoupled (\wDDc) and \windecoupled (\wDDd). \wincoupled implements a single \access \window for all data structure operations, whereas \windecoupled implements an independent \window for each data structure operation.
\subsection{\wincoupled}
\begin{algorithm*}
\SetAlgoVlined 
\SetAlgoSkip{}
\SetNlSty{}{}{} 

\caption{Window Coupled (2Dc)}
\label{algo:DDc}

\begin{multicols}{2}[]
\struct{Descriptor Des}
{                              
	*item\;                 \label{line1:struct1}
	count\;
	version\;               \label{line1:struct2}
}
\struct{Window Win}
{
	max\;
	version\;
} 
\Fn{Window(Op,index,contention)}
{
	IndexSearch $\gets$ Random $\gets$ notempty $\gets$ 0\;         \label{line1:serachentry}
	LWin $\gets$ Win\;                                              \label{line1:winlocal}
	\If{contention = True}
	{
	    index $\gets$ RandomIndex(); contention $\gets$ False\;           \label{line1:contselect}
	}
	\While{True}
	{	
		\If{IndexSearch = \width}
		{
			SHIFTWINDOW()\;
		}
		
		Des $\gets$ Array[index]\;                              \label{line1:readdesciptor}
		\uIf{Op = put $\land$ Des.count $<$ Win.max}
		{                                                       \label{line1:checkcount1}
			return \{Des,index\}\;                              \label{line1:returndescriptor1}
		}
		\uElseIf{Op = get $\land$ Des.count $>$ (Win.max - depth)}
		{                                                       \label{line1:checkcount2}
		    return \{Des,index\}\;                              \label{line1:returndescriptor2}
		}
		\uElseIf{LWin = Win}
		{                                       \label{line1:winshiftcheck1} 
		   HOP()\;
		}
		\Else
		{
		    LWin $\gets$ Win; IndexSearch $\gets$ 0\;      \label{line1:changewindow1}
		}
	}
}
\columnbreak
\Macro{SHIFTWINDOW()}
{
    \If{Op = get $\land$notempty = 0}
	{
		return \{Des,index\};                       \label{line1:returnempty}
	}
	\If{LWin = Win }
    {                                               \label{line1:winshiftcheck2}
	    \uIf{Op = put}
	    {                                           \label{line1:shiftup}
	        NWin.max $\gets$ LWin.max + ShiftUp\;
	    }
		\ElseIf{Op = get $\land$ Win.max $>$ depth}
		{                                           \label{line1:shiftdown}
		    NWin.max $\gets$ LWin.max - ShiftDown\;
		}
		NWin.version $\gets$ LWin.version + 1\;
		CAS(Win,LWin,NWin)\;                        \label{line1:shiftwindow}
	}
	LWin $\gets$ Win; IndexSearch $\gets$ 0\;       \label{line1:changewindow2}
}
\Macro{HOP()}
{
    \eIf{Random$<$2}
	{
		index $\gets$ RandomIndex(); 
		Random += 1\;                               \label{line1:randomhops}
	}
    {
        \If{Op = get $\land$ Des.item != NULL}
        {
            notempty $\gets$ 1\;                    \label{line1:nullcheck}
        }
        \eIf{index = (width - 1)}
        {
            index $\gets$ 0\;
        }
        {
            index += 1\;                            \label{line1:robinhops}
        }
         IndexSearch += 1\;
    }
}
\end{multicols}
\end{algorithm*}

\wincoupled couples both \putop and \getop to share the same \window and \localcounter for each \substructure. A successful \putop increments whereas, a successful \getop decrements the given \localcounter. On a full \window, \putop increments \winmax \shifting the \window up by a given value (\shiftup), whereas, on an empty \window, \getop decrements \winmax, \shifting the \window down by a given value (\shiftdown). \wincoupled resembles elimination \cite{shavit1995elimination}, only that here, we cancel out operation counts for matching \putop and \getop on the same \substructure within the same \window. Just like elimination reduces joint \access updates, \wincoupled reduces \window \shift updates.

In Algorithm \ref{algo:DDc}, we present the algorithmic steps for \wincoupled. Recall, $\width = \#\substructures$ and $\depth=\winmax-\winmin$. Each \substructure is uniquely identified by a descriptor that holds information including a pointer to the \substructure, \localcounter counter, and a version number (Line \ref{line1:struct1} to \ref{line1:struct2}). The version number is to avoid ABA\footnote{The ABA problem occurs when multiple processors access a shared location without noticing each others changes.} related issues, where different \substructure states might hold same descriptor update values. Using a wide \caeo, we update the descriptor information in a single atomic step (Line \ref{line1:shiftwindow}). This helps maintain the integrity of the underlying \substructures with minimal structural modifications. We use an array of descriptors through which \processors can access any of the available \substructures. 

To perform an operation, the \processor has to search and select a valid \substructure within a \window period. At the beginning of every search, the \processor stores a copy of the current \window locally (Line \ref{line1:winlocal}) which is used to detect \window \shifts while searching (Line \ref{line1:winshiftcheck1}, \ref{line1:winshiftcheck2}). During the search, the \processor validates each \substructure count against \winmax (Line \ref{line1:checkcount1}, \ref{line1:checkcount2}). If no valid \substructure is found, \winmax is updated atomically, \shifting the \window up or down (Line \ref{line1:shiftwindow}). \putop increments \winmax to \shift the \window up (Line \ref{line1:shiftup}), whereas, \getop decrements \winmax to \shift the \window down (Line \ref{line1:shiftdown}). Before index hopping or \window \shifting, the \processor must verify that the \window has not yet \shifted (Line \ref{line1:winshiftcheck1} or \ref{line1:winshiftcheck2} respectively). For every \window \shift during the search, the \processor restarts the search with the new \window (Line \ref{line1:changewindow1}, \ref{line1:changewindow2}). 

If a valid \substructure is selected, the respective descriptor state (Line \ref{line1:readdesciptor}) and index are returned  (Line \ref{line1:returndescriptor1}, \ref{line1:returndescriptor2}). The \processor can then proceed to try and operate on the given \substructure using the respective descriptor information. As a way to check if the data structure is empty, the \window search can only return an empty \substructure (Line \ref{line1:returnempty}), if during the search, all \substructures where empty ($NULL$ pointer) (Line \ref{line1:nullcheck}). Using the \window parameters, \width, and \depth, we can tightly bound the \relaxation behaviour of the derived 2Dc data-structure as discussed later in Section \ref{sec:modelcorrectness}.

However, we should note that, \window and \substructure updates occur independent of each other. For \wincoupled, this can lead to a \substructure being updated although the \window it was selected from has since \shifted. Take as an example, a \getop selecting a \substructure from a full \window ($w_1$) at time ($t_1$), followed by a \putop that reads the full \window and \shifts it up to $w_2$ at $t_2$. It is possible for the \getop to update the selected \substructure at $t_3$ based on $w_1$ that has since \shifted to $w_2$. This difference is however bounded as proved in Section \ref{sec:modelcorrectness}.
\subsection{\windecoupled}
\label{app:windecoupled}
\begin{algorithm*}
\SetAlgoVlined 
\SetAlgoSkip{}
\SetNlSty{}{}{} 

\caption{Window Decoupled (2Dd)}
\label{algo:DDd}

\begin{multicols}{2}[]
\struct{Descriptor Des}
{
	*item\;
	getcount\;                                      \label{line2:getscount}
	putcount\;                                      \label{line2:putscount}
}
\struct{Window Win}
{
	max\;
}
\Fn{Window(Op,index,contention)}
{
    \If{contention = True}
    {
        index $\gets$ RandomIndex(); contention=False\;   \label{line2:contselect}
    }
    \eIf{op = put}
    {
        return PutWindow(index)\;                   \label{line2:putwindow}
    }
    {
        return GetWindow(index)\;                   \label{line2:getwindow}
    }
}
\Fn{PutWindow(index)}
{
	IndexSearch $\gets$ Random $\gets$ 0\;          \label{line2:serachentry}
	LpWin $\gets$ pWin\;           \label{line2:winlocal1}
	\While{True}
	{	
		\uIf{IndexSearch = \width}
		{
			\If{LpWin = pWin}
			{                                                      \label{line2:winshiftcheck1}
			    NWin.max $\gets$ LpWin.max + \depth\;              \label{line2:adddepth1}
			    CAS(pWin,LpWin,NWin)\;                             \label{line2:shiftwindow1}
			}
			LpWin $\gets$ pWin; IndexSearch $\gets$ 0\;
		}
		
		Des $\gets$ putArray[index]\;                                \label{line2:readdesciptor1}
		\uIf{Des.putcount $<$ pWin.max}
		{                                                       \label{line2:checkcount1}
		    return \{Des,index\}\;                              \label{line2:returndescriptor1}
		}
		\uElseIf{LpWin = pWin}
		{                                                       \label{line2:winshiftcheck2}
		    HOP()\;
		}
		\Else
		{
		    LpWin $\gets$ pWin; IndexSearch $\gets$ 0\;         \label{line2:changewindow1}
		}
	}
}
\Fn{GetWindow(index)}
{
	IndexSearch $\gets$ Random $\gets$ notempty $\gets$ 0; LgWin $\gets$ gWin\;   \label{line2:winlocal2}
	\While{True}
	{	
		\uIf{IndexSearch = \width}
		{
			\If{LgWin = gWin}
			{                                                   \label{line2:winshiftcheck3}
			    NWin.max $\gets$ LgWin.max + \depth\;           \label{line2:adddepth2}
			    CAS(gWin,LgWin,NWin)\;                          \label{line2:shiftwindow2}
			}
			LgWin $\gets$ gWin\; 
			IndexSearch $\gets$ notempty $\gets$ 0\;
		}
		Des $\gets$ getArray[index]\;                                \label{line2:readdesciptor2}
		\uIf{Des.getcount $<$ gWin.max $\land$ Des.item != NULL}
		{                                                        \label{line2:checkcount2}
		    return \{Des,index\}\;                               \label{line2:returndescriptor2}
		}
		\uElseIf{LgWin = gWin}
		{                                                       \label{line2:winshiftcheck4}
		    \If{Des.item != NULL}
		    {
		        notempty $\gets$ 1\;                            \label{line2:nullcheck}
		    }
		    HOP()\;
		}
		\Else
		{
		    LgWin $\gets$ gWin; IndexSearch $\gets$ notempty $\gets$ 0\;    \label{line2:changewindow2}
		}
	}
}
\Macro{HOP()}
{
    \eIf{Random $<$ 2}
	{
		index $\gets$ RandomIndex(); Random += 1\;          \label{line2:randomhops}
	}
    {
        IndexSearch += 1\;
        \uIf{IndexSearch = \width $\land$ notempty = 0}
		{
		    return \{Des,index\}\;                          \label{line2:returnempty}
		}
		index += 1\;                                    \label{line2:robinhops}
        \If{index = width}
        {
            index $\gets$ 0\;
        }
    }
}
\end{multicols}
\end{algorithm*}

\windecoupled decouples the data structure operations (\putop and \getop) and assigns them independent \windows. Also, an independent \localcounter is maintained for each data structure operation, on each \substructure. Unlike \wincoupled, each successful operation increments its respective \localcounter and \winmax on a full \window. This implies that both \localcounter and \window counters increase monotonically as shown in Algorithm \ref{algo:DDd}. 

Recall that operations are decoupled to operate within independent \windows (Line \ref{line2:putwindow}, \ref{line2:getwindow}). Therefore, the \substructure descriptor differs from that used for \wincoupled; it includes a pointer to the \substructure, \localgetcounter counter for \getop operations (Line \ref{line2:getscount}) and \localputcounter counter for \putop operations (Line \ref{line2:putscount}). Note that we do not need a version number since the counter updates monotonically increment it, which avoids the ABA problem. The descriptor can be updated in one atomic step using a wide \caeo instruction (Line \ref{line2:shiftwindow1}, \ref{line2:shiftwindow2}), just like in \wincoupled. Unlike \wincoupled, \windecoupled can use one or more arrays of descriptors where necessary. As an example, we describe the \windecoupled using two arrays of descriptors through which \processors can access any of the available \substructures; \putop descriptor's for \processors performing \putop operations, and \getop descriptor's for \processors performing \getop operations. Apart from accessing a different array for each operation, the \window search steps are similar to those of \wincoupled. 


To perform an operation, the \processor has to search and select a valid \substructure within a \window period for the given operation. The type of operation determines which descriptor array to access (Line \ref{line2:readdesciptor1}, \ref{line2:readdesciptor2}) and \window to operate with (Line \ref{line2:putwindow} or \ref{line2:getwindow}). At the beginning of every search, the \processor stores a copy of the \window locally (Line \ref{line2:winlocal1}, \ref{line2:winlocal2}) which is used to detect \window \shifts while searching (Line \ref{line2:winshiftcheck1}, \ref{line2:winshiftcheck2}, \ref{line2:winshiftcheck3}, \ref{line2:winshiftcheck4}). During the search, the \processor validates each \substructure against the given \winmax specific to the operation (Line \ref{line2:checkcount1}, \ref{line2:checkcount2}).

If no valid \substructure is found, the \winmax is incremented, \shifting the \window up. Both \putop and \getop \shift the \window up (Line \ref{line2:shiftwindow1}, \ref{line2:shiftwindow2}) by incrementing the \winmax count by \depth (Line \ref{line2:adddepth1}, \ref{line2:adddepth2}). Before \window \shifting or \hopping to another index, the \processor has to confirm that the \window has not \shifted from the locally known state (Line \ref{line2:winshiftcheck2}, \ref{line2:winshiftcheck4}). This guarantees that \processors always start their search within the most current \window for the respective operation. For every \window \shift during the search, the \processor restarts the search with the new \window (Line \ref{line2:changewindow1}, \ref{line2:changewindow2}).

If a valid \substructure is selected, the respective descriptor state (Line \ref{line2:readdesciptor1}, \ref{line2:readdesciptor2}) and index are returned  (Line \ref{line2:returndescriptor1}, \ref{line2:returndescriptor2}). The \processor can then proceed to try and operate on the given \substructure using the descriptor information. As a way to check if the data structure is empty, the \window search can only return an empty \substructure (Line \ref{line2:returnempty}), if during the search, all \substructures where empty ($NULL$ pointer) (Line \ref{line2:nullcheck}). Just like \wincoupled, we use the \window parameters, \width, and \depth, to tightly bound the \relaxation behaviour of the derived 2Dd data-structures as we will discuss in Section \ref{sec:modelcorrectness}.

\section{Optimisations}
\label{sec:optimization}
Our multi-structure 2D design framework can be tuned to optimise for; locality, contention and \substructure search overhead (\hopos), using the \width and \depth parameters. 

Locality occurs when a \processor accesses and operates on the same valid \substructure consecutively before another \processor operates on the given \substructure. This implies that, the data corresponding to the \accessp of the given \substructure will mostly be maintained in the local memory of the \processor operating on the \substructure for the given number of consecutive operations.

Contention occurs when more than one \processor try to access and operate on the same valid \substructure at the same time, in other words, \processors will contend for the given valid \substructure. In our case where synchronisation is implemented using CAS, only one of the \processors will succeed to operate on the \substructure, while the other contending \processors will have to retry. Each \processor maintains a contention variable (\codetxt{contention}) that is set to true when a \processor fails to operate on a valid \substructure due to a failed CAS, signalling the presence of contention.

A \textit{\hopo} occurs when a \processor tries to access an invalid \substructure. In this case, the \processor has to \hopo from the invalid \substructure to another \substructure while searching for a valid \substructure. \Hopos increase as the number of invalid \substructures increases. Note that this is different from when a \processor fails on a given \substructure due to contention and retries by searching for another valid \substructure.

\subsection{Locality} 
To exploit locality, the \processor starts its \substructure search from the previously known \substructure index on which it succeeded. This allows the \processor a chance to operate on the same \substructure multiple times locally, given that the \substructure is still valid. To further improve on locality, a \processor that fails on a selected \substructure due to contention (failed CAS), randomly searches for another \substructure (Line \ref{line1:contselect}, \ref{line2:contselect}) leaving the successful \processor to take over the selected \substructure locally. The \processor will operate on the same \substructure locally, for as long as it does not fail due to contention and the \substructure is valid. Using the \depth parameter, we can control the number of consecutive data structure operations that can be performed on a \substructure for a given \window period, consequently allowing us to tune locality. 

Working locally improves the caching behaviour of the system by reducing remote memory access, consequently improving throughput performance \cite{Hackenberg:2009:CCA:1669112.1669165,David:2013:EYA:2517349.2522714,Schweizer:2015:ECA:2923305.2923811}. Locality also reduces the number of \hopos as discussed in Section \ref{sec:complexity} and supported by Theorem \ref{th:per1} and \ref{th:per2}. The cost of \hopping includes reading new memory locations, which introduces higher memory latency and cache coherence costs. Reducing the number of \hopos improves throughput performance especially under a NUMA execution environment with high communication cost between NUMA nodes \cite{Hackenberg:2009:CCA:1669112.1669165,David:2013:EYA:2517349.2522714,Schweizer:2015:ECA:2923305.2923811}. 

\subsection{Contention}
A failed operation on a valid \substructure signals the possibility of contention. The \processor that fails on a valid \substructure, starts the \substructure search on a randomly selected index (Line \ref{line1:contselect}, \ref{line2:contselect}). This reduces possible contention that might arise if the failed \processors were to retry on the same \substructure. Furthermore, random selection avoids contention on individual \substructures by uniformly redistributing the failed \processors to all available \substructures.

At the beginning of the \substructure search (Line \ref{line1:serachentry}, \ref{line2:serachentry}, \ref{line2:winlocal2}), if the search start index has an invalid \substructure, the \processor tries a given number of uniformly distributed random \hopos (Line \ref{line1:randomhops}, \ref{line2:randomhops}) before switching to round robin \hopos (Line \ref{line1:robinhops}, \ref{line2:robinhops}) until a valid \substructure is found. The random \hopos evenly distributes the \processors to avoid contention and when a valid \substructure is not found in the initial random \hopos, \processors end up with randomly distributed starting points for the round robin search, further avoiding contention that can arise from \processors contending on the same \substructures during the search. In our implementation we use two random \hopos as the optimal number for a random search, based on the power of random two choices result \cite{mitzenmacher2001power}. However, this is a configurable parameter that can take any value. 

We further note that contention is inversely proportional to the \width. As a simple model, we split the latency of an operation into a contention ($op_{cont}$) and a contention-free ($op_{free}$) operation cost, estimated by the following formula: $op= op_{cont}/\width + op_{free}$. This means that we can increase the \width to further reduce contention when necessary.

\subsection{\Hopos}
The number of \hopos increases with an increase in \width due to the increase in possible numbers of invalid \substructures. This counteracts the performance benefits from contention reduction through increasing \width, necessitating a balance between contention and \hopos reduction. Based on our simple contention model above ($op= op_{cont}/\width + op_{free}$), the performance would increase as the contention factor vanishes with the increase of $\width$, but with an asymptote at $1/op_{free}$. This implies that beyond some point, one cannot really gain throughput by increasing the \width, however, throughput would get hurt due to the increased number of \hopos. At some point as \width increases, gains from the contention factor ($\lim_{\width \to \infty} op_{cont} \to 0$) are surpassed by the increasing cost of \hopos. This is something that we want to avoid in our effort of \relaxing semantics for gaining throughput. To avoid this, we switch to increasing \depth instead of \width, at the point of \width saturation. Increasing \depth reduces the number of \hopos. This is supported by our step complexity analysis presented in Section \ref{sec:complexity} Theorem \ref{th:per1} and \ref{th:per2}.

\wincoupled can further be optimised to minimise the number of search \hopos that arise from \shifting the \window. A \processor performing a \putop operation will \shift the \window up if all available \substructures have a maximum number of operation count. On the other hand, a \processor performing a \getop operation will \shift the \window down if all available \substructures have a minimum number of operation count given that the data structure is not empty. \Shifting the \window in intervals of \depth ($\fshift = \depth$), \shifts the \window up to an empty state where all \substructures will have the minimum operation count for the given \window, or \shifts the \window down to a full state where all \substructures will have the maximum operation count for the given \window. \Shifting the \window up invalidates all \substructures for \getop operations at that time (Line \ref{line1:checkcount2}), while \shifting the \window down invalidates all \substructures for \putop operations at that time (Line \ref{line1:checkcount1}). To reduce \substructure invalidation on \window \shifts and subsequently minimise the number of \hopos on both full and empty \window states, we configure the \window to \shift with respect to the operation rate. This implies that, if $\pusho rate > \popo rate$ then $\fshift>\bshift$ where $\fshift+\bshift=\depth$. With this configuration, the \window will \shift with a bigger interval for the operation with a higher rate, giving \processors performing that operation a higher chance to find a valid \substructure without invalidating \substructures for the other operation. This reduces \hopos and also enhances locality.

\section{Complexity Analysis}
\label{sec:complexity}
In this section, we analyse the correlation between \hopos and the \window parameters \width and \depth. We provide the expected step complexity of our framework by considering the sequential executions of our algorithm; where a single \processor executes a sequence of operations. The type of the operations in the sequence is determined independently with a fixed probability, where $p$ denotes the probability of a \putop operation. 
\subsection{\wincoupled}
\label{sec:DDccomplexity}
Recall \winmax regulates the maximum number of operations per \substructure within a given \window. Also recall that \width = \#\substructures. Let the number of operations of a \substructure $i$ be given by \sjitems. \putop operations are allowed to occur at \substructure $i$, if $\sjitems \in [\globalcounter - \depth, \globalcounter-1]$, whereas \getop operations are allowed to occur at \substructure $i$, if $\sjitems \in [(\globalcounter-\depth)+1, \globalcounter]$. This basically means that, at any time, the number of operations of a \substructure can only variate in the vicinity of \globalcounter, more precisely: $\forall i, (\globalcounter-\depth) \leq  \sjitems \leq \globalcounter$. 
We refer to this interval as the active region of the \substructure.

We introduce the random variable $\sjwitems = \sjitems - (\globalcounter-\depth)$ where $\sjwitems \in [0,\depth]$ that
provides the number of items in the active region of the \substructure $i$ and the random variable $\witems=\sum_{i=1}^{\width} \sjwitems$ that provides the total number of items in the \window.

As mentioned before the \depth dimension tries to exploit locality, thus, a \processor starts an operation with a query on the \substructure where the given \processor's last successful operation occurred. This means that a \processor will \hopo iff $\sjwitems=\depth$ for a \putop operation, whereas for a \getop operation, the \processor will \hopo iff $\sjwitems=0$. Therefore, the number of \substructures, whose active regions are full, is given by $\floor{(\witems/\depth)}$ at a given time, because the \processor does not leave a \substructure until its active region gets either full or empty. If the \processor has to \hopo a \substructure, then a new \substructure is selected uniformly at random from the remaining set of \substructures. If none of the \substructures fulfills the condition (implies that $\witems = 0$ for a \getop, or, $\witems=\depth \times \width$ for a \putop), then the \window \shifts based on a given \shift parameter (\ie for a \putop operation $\globalcounter=\globalcounter+\fshift$ and for a \getop operation $\globalcounter=\globalcounter-\bshift$, where $1 \leq \bshift,\fshift \leq \depth$). One can observe that the value of \witems before an operation defines the expected number of \hopos and the \shift of the \window.

To compute the expected step complexity of an operation that occurs at a random time, we model the random variation process around the \globalcounter with a Markov chain, where the sequence of \putop and \getop operations lead to the state transitions. We denote the \window \shift parameter as $\shift$, where $\shift = \fshift = \bshift = \depth$. As a remark, we consider the performance of the \substructures mostly when they are non-empty, since \getop($NULL$) and \putop would have no \hopos in this case. The Markov chain is strongly related to \witems. It is composed of $\wsize+1$ states $\sta{0}, \sta{1}, \dots, \sta{\wsize}$, where $\wsize=\depth \times \width$.
For all $i \in \inte{0}{\wsize}$, the operation is in state \sta{i} iff $\witems=i$.
For all $(i,j) \in \intedef{\wsize+1}^2$, $\pro{\sta{i} \rightarrow \sta{j}}$ denotes the state transition probability, that is given by the following function, where $p$ denotes the probability of a \putop:

$$
\begin{cases}
\pro{\sta{i} \rightarrow \sta{i+1}}=p, & \text{if } 0<i<\wsize \\
\pro{\sta{i} \rightarrow \sta{i-1}}=1-p, & \text{if } 0<i<\wsize \\
\pro{\sta{i} \rightarrow \sta{\wsize-(\shift \times \width + 1)}}=p, & \text{if } i=\wsize\\
\pro{\sta{i} \rightarrow \sta{(\shift \times \width - 1)}}=1-p, & \text{if } i=0\\
\pro{\sta{i} \rightarrow \sta{j}}=0, & \text{otherwise}
\end{cases}
$$

The stationary distribution (denoted by the vector $\pi = (\pi_i)_{i\in \inte{0}{\wsize}}$) exists for the Markov chain
, since the chain is 
irreducible and positive recurrent (note that state space is finite). The left eigenvector of the transition matrix with eigenvalue $1$ provides the unique stationary distribution.

\begin{lemma}
\label{lemma:per1}
For the Markov chain that is initialized with $p=1/2$ and $\shift$, where $l=\shift \times \width -1$, the stationary distribution is given by the vector $\pi^l=(\pi^l_0 \pi^l_1 .. \pi^l_\wsize)$, assuming $ \wsize-l >= l$ (for $l>K-l$, one can obtain the vector from the symmetry $\pi^l = \pi^{\wsize-l}$): 
(i) $\pi^l_i=\frac{i+1}{(l+1)(\wsize+1-l)}, \text{ if } i < l$; (ii) $\pi^l_i=\frac{l+1}{(l+1)(\wsize+1-l)}, \text{ if } l \leq i \leq \wsize-l$; 
(iii) $\pi^l_i=\frac{\wsize-i+1}{(l+1)(\wsize+1-l)}, \text{ if } i > \wsize-l$.

\end{lemma}
\begin{proof}
We have stated that the stationary distribution exist since the chain is aperiodic and irreducible for all $p$ and $\shift$.\\ Let $(M_{i,j})_{(i,j) \in \intedef{\wsize}^2}$ denote the transition matrix for $p=1/2$ and $\shift$. The stationary distribution vector $\pi^l$ fulfills, $\pi^l M = \pi^l$, that provides the following system of linear equations:
(i) $2\pi^l_0 = \pi^l_1$; (ii) $2\pi^l_\wsize = \pi^l_{(\wsize-1)}$; (iii) $2\pi^l_i = \pi^l_{i-1} + \pi^l_{i+1}$; (iv) $2\pi^l_l = \pi^l_{i-1} + \pi^l_{i+1} + \pi^l_1$; (v) $2\pi^l_{\wsize-l} = \pi^l_{i-1} + \pi^l_{i+1} + \pi^l_\wsize$. 

In case, $l=\wsize-l$, then (iv) and (v) are replaced with $2\pi^l_{(l=\wsize-l)} = \pi^l_{i-1} + \pi^l_{i+1} + \pi^l_1 + \pi^l_\wsize$.

Based on a symmetry argument, one can observe that, for all $l$, $\pi^l_i=\pi^l_{\wsize-i}$ the system can be solved in linear time ($O(\wsize)$) by assigning any positive (for irreducible chain $\pi^l_i > 0$) value to $\pi^l_0$. The stationary distribution is unique thus for any $\pi^l_0$, $\pi^l$ spans the solution space. We know that $\sum_{i=0}^{K} \pi^l_i=1$, starting from $\pi^l_0=1$, we obtain and normalize each item by the sum.
\end{proof}

An operation starts with the search of an available \substructure. This search contains at least a single query at the \substructure where the \processor's last successful operation occurred, the rest incur a \hopo step. In addition, the operation might include the \shift of the window, as an extra step, denoted by \slideop. We denote the number of extra steps with $Extra=\hopo + \slideop$. With the linearity of expectation, we obtain $\expe{Extra} = \expe{\hopo} + \expe{\slideop}$. 
Relying on the law of total expectation, we obtain:\\ (i) $\expe{\hopo} = \sum_{i=0}^{\wsize} \sum_{\treiberop \in \{pop, push\}} \expe{\hopo \vert \sta{i}, \treiberop} \pro{\sta{i}, \treiberop}$;\\ 
(ii) $\expe{\slideop} = \sum_{i=0}^{\wsize} \sum_{\treiberop \in \{pop, push\}} \expe{\slideop \vert \sta{i}, \treiberop} \pro{\sta{i}, \treiberop}$;\\
where $\pro{\sta{i}, \treiberop}$ denotes the probability of an operation to occur in state $\sta{i}$. We analyze the 
algorithm for the setting where $\shift=\depth$ and $p=1/2$. We do this because the bound, that we manage to find in this case, is tighter, and gives a better idea of the influence of the $2D$ parameters to the expected performance. 
For this case the stationary distribution is given by Lemma~\ref{lemma:per1}.

\begin{theorem}
\label{th:per1}
For a \DDc that is initialized with parameters \depth, \width, $\shift=\depth$ and $p=1/2$, $\expe{Extra}=O(\frac{\ln \width}{\depth})$.
\end{theorem}

\begin{proof}

Firstly, we consider the expected number of extra steps for a \putop operation. Given that there are \witems items, a \putop attempt would generate an extra step if it attempts to add and item to a \substructure that has $\sjwitems = \depth$ items. 
Recall that the \processor sticks to a \substructure until it is not possible to conduct an operation on it. This implies that the extra steps can be taken only in the states $\sta{i}$ such that $i (mod\text{ }\depth)=0$, because the \processor does not leave a \substructure before $\sjwitems = 0$ or $\sjwitems = \depth$. In addition, a \putop (\getop) can only experience an extra step if the previous operation was also a \putop (\getop).

Given that we are in $\sta{i}$ such that $i (mod\text{ }\depth)=0$, then the first requirement is to have a \putop as the previous operation. If this is true, then the \putop operation has to \hopo to another \substructure, which is selected from the remaining set of \substructures uniformly at random. At this point, there are $f=\frac{i}{\depth}-1$ full \substructures in the remaining set of \substructures. If a full \substructure is selected from this set, this leads to another \hopo and again a \substructure is selected uniformly at random from the remaining set of \substructures.

Consider a full \substructure (one of the $f$), this \substructure would be \hopped if it is queried before querying the \substructures that are empty. There are $\width-f-1$ empty \substructures, thus a \hopo in this \substructure would occur with probability $1/(\width-f)$. There are $f$ such \substructures. With the linearity of expectation, the expected number of \hopos is given by: $f/(\width-f)+1=\width/(\width-f)$. Which leads to $\expe{\hopo \vert \sta{i}, \putop} = p \times \width/(\width-f)$ if $i (mod\text{ }\depth)=0$ or $\expe{\hopo \vert \sta{i}, \putop}=0$ otherwise.

From Lemma~\ref{lemma:per1}, $\pi_i < 2/(\wsize+1)$ we obtain:\\ 
\begin{align*}
\expe{\hopo \vert \putop}
= \sum_{i=0}^{\wsize} \pi_i \expe{\hopo \vert \sta{i}, \putop}
< (\sum_{f=0}^{\width-1} \frac{\width}{\width-f}) \frac{2p}{\wsize+1}
< (\ln (\width-1) + \gamma) \frac{\width}{\wsize+1}
< (\ln \width + \gamma) \frac{1}{\depth}
\end{align*}

The bounds for $\expe{\hopo \vert \putop}$ also hold for $\expe{\hopo  \vert \getop}$. Given that there are $\wsize-i$ (system is in state $\sta{i}$) empty \substructures then there are $e=\floor{\frac{\wsize-i}{\depth}}-1$ \substructures whose \window regions are empty, minus the \substructure that the \processor last succeeded on. Using the same arguments that are illustrated above (replace $f$ with $e$ and p=1-p), we obtain the same bound.

\Window only \shifts at $\sta{\wsize}$ if a \putop operation happens and at $\sta{0}$ if a \getop operation happens. Hence: $\expe{\slideop} < \frac{2}{\wsize+1} p + \frac{2}{\wsize+1} (1-p)$.
Finally, using $\expe{Extra}=\expe{\hopo}+\expe{\slideop}$ we obtain the theorem.
\end{proof}
\subsection{\windecoupled}
\label{sec:DDdcomplexity}
Now, we apply the same reasoning for the \wincoupled to analyse \windecoupled behaviour. There are two \windows, let  \globalcounterEnq and \globalcounterDeq represent \winmax for \putop and \getop respectively. The two counters; \globalcounterEnq and \globalcounterDeq, increase monotonically.
\putop and \getop have the same complexity since they apply 
the same \window strategy with the only difference that one consumes and the other produces items. Therefore, we analyse only \putop.
Let $\sjwitems$ denote the number of items in the active region of 
the \substructure $i$ for \putop. 
We consider $\shift =\depth$ since the \globalcounterEnq is monotonically increasing. 

We again model the process with a Markov chain where the states are strongly related to $\witems=\sum_{i=1}^{\width} \sjwitems$.
It is composed of $\wsize$ states $\sta{1}, \dots, \sta{\wsize}$, 
where $\wsize=\depth \times \width$.
For all $(i,j) \in \intedef{\wsize}^2$, $\pro{\sta{i} \rightarrow \sta{j}}$ denotes the state transition probability, that is given by the following function: 
$\pro{\sta{i} \rightarrow \sta{i+1}}=1, \text{if } 1 \leq i \leq \wsize-1$ and $\pro{\sta{i} \rightarrow \sta{1}}=1, \text{if } i = \wsize$. 
The stationary distribution is given by the vector $\pi^l=(\pi^l_1 \pi^l_2 .. \pi^l_\wsize)$, where $\pi_i=1/\wsize$.

\begin{theorem}
\label{th:per2}
For a \DDd that is initialised with parameters \depth, \width, $\shift=\depth$, $\expe{Extra}=O(\frac{\ln \width}{\depth})$.
\end{theorem}
\begin{proof}
We consider the expected number of extra steps for an \putop that 
would generate an extra step if it attempts on a \substructure that has $\sjwitems = \depth$ items. 
Recall that the \processor sticks to a \substructure until it is not possible to conduct an operation on it, thus extra steps will 
be taken only in the states $\sta{i}$ such that $i (mod\text{ }\depth)=0$ and before the first hop, 
there are $f=\frac{i}{\depth}-1$ full \substructures in the remaining set of \substructures. Plugging, $\pi_i = 1/(\wsize)$ into the 
reasoning that is provided in Theorem~\ref{th:per1}, we obtain the theorem.
\end{proof}
\section{Deriving 2D Data structures}
\label{sec:algorithms}
In this Section, we show how our framework can be used to derive \koo data structures. Using \wincoupled we derive a \sDDc and a \cDDc, whereas by using \windecoupled, we derive a \sDDd, a \qDDd, a \cDDd and a \dDDd as shown in Table \ref{table:2Dalgorithms}. The base algorithms include but not limited to; Treiber's stack \cite{Treiber1986systems}, MS-queue \cite{Michael:1996:SFP:248052.248106} and Maged Deque \cite{DBLP:conf/europar/Michael03} for Stack, FIFO Queue and Deque respectively. In order to be precise and avoid reciting the base algorithms' information, we focus our discussion on the algorithmic modifications made to fit our design framework.  

\begin{table}
  \centering
  \begin{tabular}{l l l }
    \hline
    Algorithm & \wincoupled & \windecoupled \\ \hline
    \sDD & \sDDc & \sDDd \\
    \qDD & --- & \qDDd  \\
    \cDD & \cDDc & \cDDd \\
    \dDD & --- & \dDDd \\
  \end{tabular}
  \caption{$2D$ derived data structure algorithms}
  \label{table:2Dalgorithms}
\end{table}
\subsection{\sDD}
\begin{algorithm*}
\SetAlgoVlined 
\SetAlgoSkip{}
\SetNlSty{}{}{} 

\caption{\sDDc}
\label{algo:sDDc}

\begin{multicols}{2}[]
\struct{item}
{
	value;
	next\;
}
\struct{Descriptor Des}
{
    item;
    count;
    version\;              \label{line3:descriptor}
}
\Fn{Push(NewItem)}
{
	contention $\gets$ False\;
	\While{True}
	{
		\{Des,index\} $\gets$ Window(push,index,contention)\;        \label{line3:wincall1}
		NewItem.next $\gets$ Des.item\;
		NDes.item $\gets$ NewItem;
		NDes.count $\gets$ Des.count + 1;
		NDes.version $\gets$ Des.version + 1\;                    \label{line3:prepdescriptor1}
		\eIf{CAS(Array[index],Des,NDes)}
		{                                               \label{line3:cae1}
			return 1\;
		}
		{
		    contention $\gets$ True\;                          \label{line3:cont1}
		}
	}
}
\Fn{Pop()}
{
	contention $\gets$ False\;
	\While{True}
    {		
		\{Des,index\} $\gets$ Window(pop,index,contention)\;      \label{line3:wincall2}
		\uIf{Des.item != NULL}                      
		{                                           \label{line3:checknull}
			NDes.item $\gets$ Des.item.next;
			NDes.count $\gets$ Des.count - 1;
			NDes.version $\gets$ Des.version + 1\;                \label{line3:prepdescriptor2}
			\eIf{CAS(Array[index],Des,NDes)}
			{                                           \label{line3:cae2}
				return Des.item\;			
			}
			{
			    contention $\gets$ True\;                       \label{line3:cont2}               
			}
		}
		{
			return Null\;                                \label{line3:returnnull}
		}
    }
}
\end{multicols}
\end{algorithm*}
A stack is characterised by two operations: \pusho that adds an item and \popo that removes an item from the stack. Our derived \sDD algorithms are composed of multiple lock-free \substaks. Each \substak is implemented using a linked-list following the Treiber's stack design \cite{Treiber1986systems}, modified only to fit the \window framework design as discussed below. 

In Algorithm \ref{algo:sDDc} we present the algorithmic implementation of \sDDc as an example of how to use \windecoupled to derive \koo data structures. The stack head is modified to a descriptor containing the top item pointer, operation count, and descriptor version (Line \ref{line3:descriptor}). Note that, the descriptor is updated in single atomic step using a wide \caeo (Line \ref{line3:cae1}, \ref{line3:cae2}), the same way as in the Treiber's stack. \sDDc is accessed through a single array of descriptors for both \pusho and \popo operations. 

To perform an operation, a given \processor obtains a \substak by performing a \window search (Line \ref{line3:wincall1}, \ref{line3:wincall2}) as discussed earlier in Section \ref{sec:2Dtechnique}. The \processor then prepares a new descriptor based on the current descriptor at the given index (Line \ref{line3:prepdescriptor1}, \ref{line3:prepdescriptor2}). Using a \caeo, the \processor tries to atomically swap the current descriptor with the new one (Line \ref{line3:cae1}, \ref{line3:cae2}). If the \caeo fails, the \processor sets the contention indicator to true (Line \ref{line3:cont1}, \ref{line3:cont2}) and restart the \window search. The contention indicator signals the presence of contention prompting the \window search to starts from an index that is selected uniformly at random (Line \ref{line1:contselect}) as discussed earlier in Section \ref{sec:2Dtechnique}.

A successful \pusho increments whereas a \pusho decrements the operation count by one (Line \ref{line3:prepdescriptor1}, \ref{line3:prepdescriptor2}). Also, the topmost item pointer is updated. At this point, a \pusho adds an item whereas a \popo returns an item for a non-empty or $NULL$ for an empty stack (Line \ref{line3:returnnull}). An empty \substak is represented by a $NULL$ item pointer within the \descriptor (Line \ref{line3:checknull}). Recall that a \processor performs a special data structure emptiness check, by scanning all the available \substaks and returning $NULL$ only if all the \substaks where found to be empty.

\windecoupled is used to derive the \sDDd. This follows the same procedure as discussed above, with the difference being that; \pusho and \popo operations increment different operation counters on success. Also, we declare a single array of descriptors through which \processors performing a \pusho operation or \popo can access any of the available \substaks. Each descriptor maintains two different counters; one for \pusho and the other for \popo. However, the \pusho and \popo operation \windows are independent because they follow different operation counts. 
Note that, we do not need version number because the counters increase monotonically. 

\subsection{\qDD}
\label{sec:qDD}
\begin{algorithm*}
\SetAlgoVlined 
\SetAlgoSkip{}
\SetNlSty{}{}{} 

\caption{\qDDd Algorithm}
\label{algo:qDDd}

\begin{multicols}{2}[]
\struct{item}
{
	value;
	next\;
}
\struct{Descriptor Des}
{
    item;
    count\;              \label{line4:descriptor}
}
\Fn{Enqueue(NewItem)}
{
	contention $\gets$ False\;
	\While{True}
	{
		\{putDes,index\} $\gets$ Window(put,index,contention)\;     \label{line4:wincall1}
		Tail $\gets$ putDes.item\;
		NDes.item $\gets$ Item\;
		NDes.putcount $\gets$ putDes.putcount + 1\;                 \label{line4:eincrement}
		\eIf{Tail.next  =  NULL}
		{
		    \eIf{CAS(Tail.next, NULL, Item)}
		    {                                               \label{line4:add1}
		        break\;
		    }
		    {
			    contention $\gets$ True\;                          \label{line4:cont1}
			}
		}
		{
    		NDes.item $\gets$ Tail.next\;
    		\If{!CAS(putArray[index],putDes,NDes)}
    		{                                           \label{line4:help1}
			    contention $\gets$ True\;                       \label{line4:cont2}
			}
		}
	}
	\If{!CAS(putArray[index],putDes,NDes)}
	{                                                   \label{line4:cae1}
		contention $\gets$ True\;                               \label{line4:cont3}
	}
	return True\;
}
\vfill\null
\columnbreak
\Fn{Dequeue()}
{
	contention $\gets$ False\;
	\While{True}
    {		
		\{getDes,index\} $\gets$ Window(get,index,contention)\;      \label{line4:wincall2}
		Head $\gets$ getDes.item\;
		putDes $\gets$ putArray[index];
		Tail $\gets$ putDes.item\;
		\eIf{Head  =  Tail}
		{
		    \eIf{Head.next  =  NULL}
		    {
		        return NULL\;                               \label{line4:returnnull}
		    }
		    {
		        NDes.item $\gets$ Tail.next\;
    		    NDes.putcount $\gets$ putDes.putcount + 1\;    \label{line4:helpincrement}
    		    \If{!CAS(putArray[index],putDes,NDes)}
    		    {                                       \label{line4:help2}
    		        contention $\gets$ True\;              \label{line4:cont5}
    		    }
		    }
		}
		{
		    NDes.item $\gets$  Head.next\;
		    NDes.count $\gets$ getDes.getcount + 1\;            \label{line4:dincrement}
			\eIf{CAS(getArray[index],getDes,NDes)}
			{                                            \label{line4:cae2}
			    return Head.next.val\;                  \label{line4:returnv}
			}
			{
			    contention $\gets$ True\;               \label{line4:cont6}
			}
		}
    }
}
\vfill\null
\end{multicols}
\end{algorithm*}
FIFO Queues are characterised by two operations, \enqop which adds an item to the queue and \deqop which removes an item. The two operations access the queue from different points; head for \deqop and tail for \enqop. We use \windecoupled to derive a \qDD, due to its ability to maintain the independent operation counts and decriptor's access arrays. \qDD is composed of multiple lock-free \subqueues. Each \subqueue is implemented using a linked list following the Michael Scott FIFO queue (\msqueue) design \cite{Michael:1996:SFP:248052.248106}, modified only to fit the \window processes as shown in Algorithm \ref{algo:qDDd}. 

The queue head and tail are modified into an independent descriptor (Line \ref{line4:descriptor}) for each. The head descriptor contains the head item pointer and the \deqop operation counter, whereas the tail descriptor contains the tail item pointer and the \enqop operation counter. The \qDD descriptor is also updated in one atomic step using a wide \caeo (Line \ref{line4:cae1}, \ref{line4:cae2}), same as in \msqueue updates. \qDD composed of two arrays of descriptors; one holds the head descriptors (\codetxt{getDes}), while the other holds the tail descriptors (\codetxt{putDes}). Each \subqueue descriptor pair is uniquely identified by the same index value in both arrays.   

To perform an operation, a given \processor obtains a \subqueue by performing a \window search on a given descriptor array (Line \ref{line4:wincall1}, \ref{line4:wincall2}) as discussed earlier in Section \ref{sec:2Dtechnique}. An \enqop completes in two steps: First, the \processor tries to add the new item to the queue list (Line \ref{line4:add1}), if successful, the \processor then tries to update the tail descriptor with the new state (Line \ref{line4:cae1}). If a \processor encounters an incomplete \enqop, the \processor can help complete the pending \enqop by updating the tail descriptor accordingly (Line \ref{line4:help1}, \ref{line4:help2}). \caeo failure during an \enqop or \deqop operation signals the presence of contention (Line \ref{line4:cont1}, \ref{line4:cont2}, \ref{line4:cont3}, \ref{line4:cont5}, \ref{line4:cont6}) on the given \subqueue. When contention is detected, the \processor starts the next \window search on an index that is selected uniformly at random (Line \ref{line2:contselect}) as discussed earlier in Section \ref{sec:2Dtechnique}.

Both \enqop and \deqop increment their respective \subqueue descriptor operation count by one on success (Line \ref{line4:eincrement}, \ref{line4:dincrement}). Note that, \deqop increments the \enqop operation count if it helps complete a pending \enqop (Line \ref{line4:helpincrement}). On a successful operation, an \enqop adds an item whereas a \deqop returns a value for a non empty \subqueue (Line \ref{line4:returnv}) or $NULL$ (Line \ref{line4:returnnull}) for an empty Queue. Similar to \sDD, the \deqop only returns $NULL$ if the search within a given \window cannot find a non-empty \subqueue after scanning all the available \subqueues. 
\subsection{2D-Deque}
\begin{algorithm*}
\SetAlgoVlined 
\SetAlgoSkip{}
\SetNlSty{}{}{} 

\caption{2Dd-Deque Algorithm}
\label{algo:dDDd}

\begin{multicols}{2}[]
\struct{Descriptor Des}
{
    *right;
    *left;
    status\;
    GetRightCount;
    GetLeftCount;
    PutRightCount;
    PutLeftCount\;              \label{line5:descriptor}
}
\struct{item}
{
    *right;
    *left;
    value\;
}
\Fn{PushLeft(NewItem)}
{
	contention $\gets$ False\;
	\While{True} 
	{	
		\{Des,index\} $\gets$ Window(putL,index,contention)\;   \label{line5:wincall1}			
		\uIf{Des.left = NULL} 
		{
			NDes $\gets$ Des;
			NDes.PutLeftCount+=1\;
			NDes.left $\gets$ NewItem;
			NDes.right $\gets$ NewItem\;
			\If{CAS(Array[index],Des,NDes)}
			{	                                        		 \label{line5:pemptycae1}
				break;
			}
			contention $\gets$ True\;                           \label{line5:cont1}
		}
		\uElseIf{Des.status = STABLE}
		{
			NewItem.right $\gets$ Des.left\;
			NDes $\gets$ Des;
			NDes.PutLeftCount+=1\;
			NDes.left $\gets$ NewItem;
			NDes.status $\gets$ LEFTPUSH\;
			\If{CAS(Array[index],Des,NDes)} 
			{			                                        \label{line5:pushstp11}                                          
				StabiliseLeft(NDes,deque);                       \label{line5:pushstabilize1}
				break\;
			}
			contention $\gets$ True;                            \label{line5:cont2}
		}
		\Else
		{
		    stabilise(Des, index);                       \label{line5:stabilize1}
		}
	}
	return 1\;
}

\Fn{PushRight(NewItem)}
{
	contention $\gets$ False\;
	\While{True} 
	{	
		\{Des,index\} $\gets$ Window(putR,index,contention)\;	        \label{line5:wincall2}		
		\uIf{Des.right = NULL} 
		{
			NDes $\gets$ Des;
			NDes.PutRightCount+=1\;
			NDes.left $\gets$ NewItem;
			NDes.right $\gets$ NewItem\;
			\If{CAS(Array[index],Des,NDes)}
			{		                                                \label{line5:pemptycae2}	
				break;
			}
			contention $\gets$ True\;                                \label{line5:cont3}
		}
		\uElseIf{Des.status = STABLE}
		{
			NewItem.left $\gets$ Des.right\;
			NDes $\gets$ Des;
			NDes.PutRightCount+=1\;
			NDes.right $\gets$ NewItem;
			NDes.status $\gets$ RIGHTPUSH\;
			\If{CAS(Array[index],Des,NDes)} 
			{			                                            \label{line5:pushstp12}
				StabiliseLeft(NDes,deque);                           \label{line5:pushstabilize2}
				break\;
			}
			contention $\gets$ True;                                \label{line5:cont4}
		}
		\Else
		{
		    stabilise(Des, index);        \label{line5:stabilize2}
		}
	}
	return 1\;
}

\columnbreak

\Fn{PopLeft()} 
{
	contention $\gets$ empty $\gets$ False\;
	\While{True} 
	{
		\{Des,index\} $\gets$ Window(getL,index,contention)\;       \label{line5:wincall3}
		\If{empty = true} 
		{
			return NULL;                                            \label{line5:returnnull1}
		}
		\uIf{Des.right = Des.left} 
		{
			\If{Des.left  !=  NULL}
		    {                                                       \label{line5:empty1}
    			NDes $\gets$ Des;
    			NDes.GetLeftCount+=1\;
    			NDes.left $\gets$ NULL;
    			NDes.right $\gets$ NULL\;
    			\If{CAS(Array[index],Des,NDes)}
    			{                                                       \label{line5:popl1}
    				item $\gets$ Des.left;			
    				break;
    			}
			}
			contention $\gets$ True\;                           \label{line5:cont5}
		}
		\uElseIf{Des.status = STABLE} 
		{
			prev $\gets$ Des.left.right;
			NDes $\gets$ Des\;
			NDes.GetLeftcount+=1;
			NDes.left $\gets$ prev\;
			\If{CAS(Array[index],Des,NDes)}
			{                                                       \label{line5:popl2}
				item $\gets$ Des.left;				
				break;
			}
			contention $\gets$ True\;                       \label{line5:cont6}
		}
		\Else
		{
		    stabilise(Des, index);                           \label{line5:stabilize3}
		}
	}
	return item;
}
\Fn{PopRight()} 
{
	contention $\gets$ empty $\gets$ False\;
	\While{True} 
	{
		\{Des,index\} $\gets$ Window(getR,index,contention)\;       \label{line5:wincall4}
		\If{empty = true}
		{
			return NULL;                                            \label{line5:returnnull2}
		}
		\uIf{Des.right = Des.left} 
		{
			\If{Des.right  !=  NULL}
		    {                                                       \label{line5:empty2}
    			NDes $\gets$ Des;
    			NDes.GetRightCount+=1\;
    			NDes.left $\gets$ NULL;
    			NDes.right $\gets$ NULL\;
    			\If{CAS(Array[index],Des,NDes)}
    			{                                                       \label{line5:popr1}
    				item $\gets$ Des.right;			
    				break;
    			}
    		}
			contention $\gets$ True\;                               \label{line5:cont7}
		}
		\uElseIf{Des.status = STABLE} 
		{
			prev $\gets$ Des.right.left;
			NDes $\gets$ Des\;
			NDes.GetRightcount+=1;
			NDes.right $\gets$ prev\;
			\If{CAS(Array[index],Des,NDes)}
			{                                                       \label{line5:popr2}
				item $\gets$ Des.right;				
				break;
			}
			contention $\gets$ True\;                           \label{line5:cont8}
		}
		\Else
		{
		    stabilise(Des, index);                                  \label{line5:stabilize4}
		}
	}
	return item;
}

\Fn{stabilise(Des, index)}
{
	\eIf{Des.status = RIGHTPUSH}
	{
	    StabiliseRight(Des,index);
	}
	{
	    StabiliseLeft(Des,index);
	}
}

\end{multicols}
\end{algorithm*}

\begin{algorithm*}
\SetAlgoVlined 
\SetAlgoSkip{}
\SetNlSty{}{}{} 

\caption{2Dd-Deque Stabiliser Algorithm}
\label{algo:stabilizedeq}

\begin{multicols}{2}[]
\Fn{StabiliseLeft(Des, index)} 
{
	\If{Array[index]!=Des}{ return; }
	prev $\gets$ Des.left.right\;
	\If{Array[index]!=Des}{ return;	}
	prevnext $\gets$ prev.left\;
	\If{prevnext != Des.left}
	{
		\If{Array[index] != Des}{ return;}
		\If{!CAS(prev.left,prevnext,Des.left)}
		{
			return\;
		}
	}
	
	NDes $\gets$ Des;
	NDes.status $\gets$ STABLE\;
	CAS(Array[index],Des,NDes);
}

\columnbreak

\Fn{StabiliseRight(Des, index)} 
{
	\If{Array[index] != Des}{ return; }
	prev $\gets$ Des.right.left\;
	\If{Array[index] != Des}{ return;	}
	prevnext $\gets$ prev.right\;
	\If{prevnext != Des.right}
	{
		\If{Array[index] != Des}{ return;}
		\If{!CAS(prev.right,prevnext,Des.right)}
		{
			return\;
		}
	}
	
	NDes $\gets$ Des;
	NDes.status $\gets$ STABLE\;
	CAS(Array[index],Des,NDes)\;
}

\end{multicols}
\end{algorithm*}
Deques are characterised by four operations, \pushleftop which adds an item to the left of the deque, \pushrightop which adds an item to the right of the deque, \popleftop which removes an item from the left of the deque, if any and \poprightop which removes an item from the right of the deque, if any. We use \windecoupled to derive a \dDD, due to its ability have multiple \windows. Each of the four \dDD operations is assigned an independent operation \window. \dDD is composed of multiple lock-free \subdeques. Each \subdeque is implemented using a doubly-linked list following the Maged (\dmaged) design \cite{DBLP:conf/europar/Michael03}, modified only to fit the framework \window design  as shown in Algorithm \ref{algo:dDDd}. 

The \subdeque is modified, replacing the anchor with a descriptor that can also be updated in one atomic step using a \caeo instruction, same as in \dmaged updates. The descriptor, holds two pointers to the leftmost and rightmost items in the \subdeque, if any, and a status tag. Four operation counters are added to the descriptor, one for each deque operation (Line \ref{line5:descriptor}). \dDD is composed of a single array of descriptors through which the \processors performing either of the four \dDD operations can access any of the available \subdeques. We use a single array because our base algorithm \dmaged uses a single access point anchor for the four deque operations. For each operation on the given side of a \subdeque, a respective operation counter is incremented by one. 

To perform an operation, a \processor obtains a \subdeque by performing a \window search respective to the operation being performed (Line \ref{line5:wincall1}, \ref{line5:wincall2}, \ref{line5:wincall3}, \ref{line5:wincall4}) as discussed earlier in Section \ref{sec:2Dtechnique}. A \pusho operation on either side is completed in one atomic step for an empty \subdeque by replacing the \subdeque descriptor with a new descriptor pointing to the new item (Line \ref{line5:pemptycae1}, \ref{line5:pemptycae2}). If the \subdeque is stable and contains one or more items, the \pusho operations complete in two steps. First step is to swing the given \subdeque descriptor pointer to the new item and to indicate unstable \subdeque in the status tag, atomically (Line \ref{line5:pushstp11}, \ref{line5:pushstp12}). After this step, the given \subdeque is unstable. The next step of a \pusho operation is to stabilise the \subdeque (Line \ref{line5:pushstabilize1}, \ref{line5:pushstabilize2}). A \popo operation on either side of a non empty \subdeque completes in one atomic step (Line \ref{line5:popl1}, \ref{line5:popl2}, \ref{line5:popr1}, \ref{line5:popr2}). For each operation, if the \subdeque is unstable, it must be stabilised (Line \ref{line5:stabilize1}, \ref{line5:stabilize2}, \ref{line5:stabilize3}, \ref{line5:stabilize4}), before attempting the given operation on the \subdeque.

\caeo failure during any of the four operations signals the presence of contention (Line \ref{line5:cont1}, \ref{line5:cont2}, \ref{line5:cont3}, \ref{line5:cont5}, \ref{line5:cont6}, \ref{line5:cont7}, \ref{line5:cont8}) on the given \subdeque. The same applies when a \deqop operation on either sides of the \subdeque encounters an empty \subdeque (Line \ref{line5:empty1}, Line \ref{line5:empty2}). When contention is detected, the \processor starts the next \window search on an index that is selected uniformly at random (Line \ref{line2:contselect}) as discussed earlier in Section \ref{sec:2Dtechnique}.
\subsection{\cDD}
\cDD is characterized by two operations; \increment (\putop) which increases the counter and \decrement (\getop) which decreases the counter. It is composed of multiple \subcounters whose local count (\localcounter) can only be greater than or equal to zero. Both \cDDc and \cDDd follow the \sDDc and \sDDd implementation details. Following the same strategy, a given successful  operation, increments or decrements a given \subcounter's \localcounter, then calculates the estimated global count value as $Count = \localcounter \times \width$. $Count$ is then returned by the \processor.
\section{Correctness}
\label{sec:modelcorrectness}
In this section, we prove the correctness of the data structures we derived in Section \ref{sec:algorithms}, including their \relaxation bounds and lock freedom. All our derived 2D data structures are linearizable with respect to \koo semantics for the respective data structure. 
\subsection{\sDDc}
\sDDc is linearizable with respect to the sequential semantics of \koo stack \cite{henzinger2013quantitative}. \sDDc \pusho and \popo linearization points are similar to those of the original Treiber's Stack \cite{Treiber1986systems}. As shown in Algorithm \ref{algo:sDDc},  \popo linearizes either by returning \nil (Line \ref{line3:returnnull}) or with a successful \caeo (Line \ref{line3:cae2}). \pusho linearizes with a successful \caeo (Line \ref{line3:cae1}).

Relaxation can be applied method-wise and it is applied only to \popo operations, that is, a \popo pops one of the topmost $k$ items. Firstly, we require some notation. The \window defines the number of operations allowed to proceed on any given \substak. The \window is \shifted by the parameter $\shift,\  1 \leq \shift \leq \depth$ and $\width=\#\substaks$. For simplicity, let  $\shift = \shiftup = \shiftdown$.
A \window $i$ ( $W_i$) has an upper bound ( $W_i^{max}$) and a lower bound ($W_i^{min}$), where $W_i^{max}=i \times \shift$ and $W_i^{min}=(i \times \shift)-\depth$, respectively. For simplicity, let \globalcounter represent the current global upper bound.  A \window is active iff $W_i^{max}=\globalcounter$. The number of items of the \substak $j$ is denoted by $N_j$, $1 \leq j \leq \width$.  To recall, the top pointer, the version number and $N_j$ are embedded into the \descriptor of \substak $j$ and all can be modified atomically with a wide \caeo instruction.

\begin{lemma}
\label{lemma:cor1}
Given that $\globalcounter=\shift \times i$, it is impossible to observe a state($S$) such that 
$N_j> W_{i+1}^{max}$ (or $N_j < W_{i-1}^{min}$).
\end{lemma}
\begin{proof}
We show that this is impossible by considering the interleaving of operations. 
Without loss of generality, assume \processor $1$ ($P_1$) has set $\globalcounter=\shift \times i$ at time $t_{1}\prime$. To do this, $P_1$ should have observed either $\globalcounter=\shift \times (i-1)$ and then $N_j=W_{i-1}^{max}$ or $\globalcounter=\shift \times (i+1)$ and then $N_j=W_{i+1}^{min}$. Let this observation of \globalcounter happen at time $t_1$. Consider the last successful push operation at \substak $j$ before the state $S$ is observed for the first time 
(we do not consider \popo operations as they can only decrease $N_j$ to a value that is less than $W_{i+1}^{max}$, this case will be covered by the first item below). 
Assume thread $0$ ($P_0$) sets $N_j$ to $N_j> W_{i+1}^{max}$ in this push operation. $P_0$ should observe $N_j \geq W_{i+1}^{max}$ and $\globalcounter > W_{i+1}^{max}$. Let $j$ be selected at time $t_0$. And the linearization of the operation happens at $t_{0}\prime>t_0$.

\begin{itemize}
\item If $t_{0}\prime < t_{1}$,  the concerned state($S$) can 
not be observed since \globalcounter cannot be changed (to $\shift \times i$)  
after $N_j > W_{i+1}^{max}$ is observed.
\item Else if $t_{1}\prime < t_{0}$, the concerned state($S$) cannot be observed since the push operation cannot proceed 
after observing \globalcounter with such $N_j$.
\item Else if $t_1 > t_0$, then $P_0$ cannot linearize because, this implies $N_j$ has been modified (the difference between the value of $Global$ that is 
observed by $P_0$ and then by $P_1$ implies this) since $P_0$ had read the 
descriptor, at least the version numbers would have changed since then.
\item Else if $t_1 < t_0$, then this implies \globalcounter has been modified, since it was read by $P_1$, thus updating \globalcounter would fail, at least based on the version number.
\end{itemize}
\end{proof}

\begin{lemma}
\label{lemma:cor2}
At all times, there exist an $i$ such that $\forall j, 1 \leq j \leq \width$: 
$W_i^{min} \leq N_j \leq  W_{i+1}^{max}$. 
\end{lemma}
\begin{proof}
Informally, the lemma states that the size (number of operations) of a \substak spans to at most two consecutive accessible  \windows. Assume that the statement is not true, then there should exist a pair of \substaks ($y$ and $z$) at some point in time such that $\exists i, N_y < W_{i}^{min}$ and $ N_z > W_{i+1}^{max}$. Consider the last \pusho at \substak $z$ and last \popo at 
\substak $y$ that linearize before or at the time $t$.

Assume thread $P_0$ (\pusho) sets $N_z$ and thread $P_1$ (\popo) sets $N_y$. To do this, $P_0$ should observe $N_z \geq W_{i+1}^{max}$ and $\globalcounter > W_{i+1}^{max}$, let \substak $z$ be selected at $t_0$. And, the linearization of the \pusho operation occurs at $t_{0}\prime>t_0$. Similarly, for $P_1$ \popo operation, let \substak $y$ be selected at $t_1$, $P_1$ should have observed $\globalcounter \leq W_{i}^{min}$. And, let the \popo operation linearize at time $t_{1}\prime>t_1$. Now, we consider the possible interleavings.

\begin{itemize}

\item If $t_{0}\prime < t_{1}$ (or the symmetric $t_{1}\prime < t_{0}$ for which we do not repeat the arguments), then for $P_1$ to proceed and pop an item from \substak $y$, it is required that $\globalcounter \leq W_{i}^{min}$. Based on Lemma~\ref{lemma:cor1}, this is impossible when $N_z > W_i^{max}$. 
\item Else if $t_1 > t_0$, then $P_0$ cannot linearize, because this implies that $N_z$ has been modified (the difference between the value of \globalcounter that is observed by $P_0$ and then by $P_1$ implies this) since $P_0$ has read the \descriptor. At least, the version number would have changed since then. 
\item Else if $t_0 > t_1$, the argument above holds for $P_1$ too, so $P_1$ should fail to linearize.
\end{itemize}
Such $N_z$ and $N_y$ pair can not co-exist at any time.
\end{proof}

\begin{theorem}

\sDDc is linearizable with respect to \koo stack semantics, where 
$k=(2\shift+\depth)(\width-1)$.

\end{theorem}
\begin{proof}
Consider the \pusho ($t_e^{push}$) and \popo ($t_e^{pop}$) linearization points, that insert and remove an item $e$ for a given \substak $j$ respectively,  where, $t_e^{pop}>t_e^{push}$. Now, we bound the maximum number of items, that are pushed after $t_e^{push}$ and are not popped before $t_e^{pop}$, to obtain $k$. Let item $e$ be the ${N_j}^{th}$ item from the bottom of the \substak. Consider a \window $i$ such that: $W_{i}^{min} \leq N_j \leq W_i^{max}$.

Lemma~\ref{lemma:cor2} states that the sizes of the \substaks should reside in a bounded region. 
Relying on Lemma~\ref{lemma:cor2}, we can deduce that at time $t_e^{push}$, the following holds: 
$\forall i: N_j \geq W_i^{min} - \shift$. Similarly, we can deduce that 
at time $t_e^{pop}$, the following holds: $\forall i: N_j \leq W_i^{max} + \shift$. 
Therefore, the maximum number of items, that are
pushed to \substak $j$ after $t_e^{push}$ and are not popped before $t_e^{pop}$ 
is at most $W_i^{max} + \shift - (W_i^{min} - \shift) = \depth + 2 \shift$. 
We know that this number is zero for \substak $j$ (the \substak that $e$ is inserted) and 
we have $\width-1$ other \substaks. So, there can be at most $(2 \shift+\depth)(\width-1)$
items that are pushed after $t_e^{push}$ and are not popped before $t_e^{pop}$. 
\end{proof}

\subsection{\sDDd}
\label{sec:sDDdcorrectness}
\begin{theorem}
\label{th:sDDd}
\sDDd is linearizable with respect to \koo stack semantics, where $k=(3\depth)(\width-1)$.
\end{theorem}

\begin{proof}

Consider the linearization points of \pusho and \popo operations that respectively insert and remove the item $e$ into and from a \substak (\substak $i$). 
Let $t_e^{push}$ and $t_e^{pop}$ denote 
these points, respectively. Now, we bound the maximum number of items, that are
pushed after $t_e^{push}$ and are not popped before $t_e^{pop}$, to obtain $k$.
We denote the number items that are pushed to (popped from) \substak $j$ 
in the time interval $[t_e^{push}, t_e^{pop}]$, with ${push}_{j}$ 
(${pop}_{j}$).

Regarding the interval $[t_e^{push}, t_e^{pop}]$, we have: (i) ${push}_{i} = {pop}_{i}$, since the number of items that are pushed into and popped from \substructure $i$ should be equal; (ii) 
$\forall j \in [1, \width], push_{j} \leq push_{i} + \depth + (\depth - (push_{i} \mod \depth))$; and (iii)
$\forall j \in [1, \width], pop_{j} \geq pop_{i} - \depth - (pop_{i} \mod \depth)$. 

Therefore, for any \substak, the number of items that are pushed after $t_e^{push}$ and are not popped before $t_e^{pop}$ ($push_{j}-pop_{j}$) is at most: $push_{i} + \depth + (\depth - (push_{i} \mod \depth)) - (pop_{i} - \depth - (pop_{i} \mod \depth)) = 3\depth$. Summing over all $j \neq i$, we obtain the theorem.
\end{proof}
\subsection{\qDDd}
\label{sec:qDDdcorrectness}

\begin{theorem}
\label{th:qDDd}
\qDDd is linearizable with respect to \koo stack semantics, where 
$k=(\depth)(\width-1)$.
\end{theorem}

\begin{proof}

The linearization points for \qDDd operations follow \msqueue design \cite{Michael:1996:SFP:248052.248106}. As shown in Algorithm \ref{algo:qDDd},  \deqop linearizes either by returning \nil (Line \ref{line4:returnnull}) or with a successful \caeo (Line \ref{line4:cae2}). \pusho linearizes with a successful \caeo (Line \ref{line4:cae1}). For readability reasons, We use  $\globalcounterEnq$ and $\globalcounterDeq$ as a representative for \enqop and \deqop \winmax respectively.
For \enqop (\deqop), \qDDd algorithm searches for a \subq whose enqueue (dequeue) counter is 
smaller than $\globalcounterEnq$ ($\globalcounterDeq$). If no such \subq exist, the $\globalcounterEnq$ ($\globalcounterDeq$)
is increased by $\depth$. 

One can observe that $\globalcounterEnq$ and $\globalcounterDeq$ are monotonically increasing. 
Also, $\globalcounterEnq= i \times \depth$ iff the \enqop counter for all \subqs are in the 
range $[(i-1) \times \depth+1, i \times \depth]$. The same holds for $\globalcounterDeq$. 
$\globalcounterEnq^i$ is updated to $\globalcounterEnq^{i+1} = \globalcounterEnq^i + \depth$
as a result of \enqop iff \enqop counter for all \subqs are equal to $\globalcounterEnq^i$ 
before the \enqop. These invariants hold at all times since the concurrent operations, that might violate 
them, would fail at \caeo instructions that modifies any \subq or any \globalcounter.

Based on these invariants, all the items that are enqueued while $\globalcounterEnq=i \times \depth$ 
will be dequeued while $\globalcounterDeq=i \times \depth$. Therefore, the maximum number 
of items that are enqueued before an item and are not dequeued before that 
item can be at most $\width \times \depth$. Disregarding the items that are 
enqueued on the same \subq, we obtain the theorem.
\end{proof}
\subsection{\dDDd}
\label{sec:dDDdcorrenctness}

\begin{theorem}
\label{th:dDDd}
\dDDd is linearizable with respect to k-out-of-order dequeue semantics, where 
$k=(8\depth)(\width-1)$.
\end{theorem}

\begin{proof}
\dDD \pusho and \popo linearization points are similar to those of the original Maged Deque \cite{DBLP:conf/europar/Michael03}. As shown in Algorithm \ref{algo:dDDd},  \popo linearizes either by returning \nil (Line \ref{line5:returnnull1}, \ref{line5:returnnull2}) or with a successful \caeo (Line \ref{line5:popl1}, \ref{line5:popl2}, \ref{line5:popr1}, \ref{line5:popr2}). \pusho linearizes with a successful \caeo (Line \ref{line5:pemptycae1}, \ref{line5:pushstp11}, \ref{line5:pemptycae2}, \ref{line5:pushstp12}).

Consider the linearization points of \pusho and \popo operations that respectively insert and remove the item $e$ into and from a \subdeque $i$. Let $t_e^{push}$ and $t_e^{pop}$ denote these points, respectively. 

First, we consider the case where \pusho and \popo operations on item $e$ occur at the same side of the \subdeque. Without loss generality, we assume they happen at the right side of the \subdeque. 
In this case, we bound the maximum number of items, that are pushed after $t_e^{push}$ from right and are not popped before $t_e^{pop}$, to obtain $k$.
We denote the number items that are right pushed to (popped from) \subdeque $j$ in the time interval $[t_e^{push}, t_e^{pop}]$, with ${push}_{j}$ 
(${pop}_{j}$). Regarding the interval $[t_e^{push}, t_e^{pop}]$, we have: (i) ${push}_{i} = {pop}_{i}$, since the number of items that are pushed into and popped from \subdeque $i$, from the same side, should be equal; (ii)  
$\forall j \in [1, \width], push_{j} \leq push_{i} + 2\depth$; and (iii) $\forall j \in [1, \width], pop_{j} \geq pop_{i} - 2\depth$. 
Summing over all $\width$ \subdeques, we have for pushes at most: $\sum_{j=0}^{\width-1} push_{j} \leq \width (push_{i}) + 2(\width-1)\depth$ and for pops at least: $\sum_{j=0}^{\width-1} pop_{j} \geq \width (pop_{i}) - 2(\width-1)\depth$. Therefore, the number of items that are pushed from right after $t_e^{push}$ and are not popped before $t_e^{pop}$ ($push_{j}-pop_{j}$) is at most: $\sum_{j=0}^{\width-1} push_{j} - \sum_{j=0}^{\width-1} pop_{j} \leq 4(\width-1)\depth$.

Second, we consider the case where \pusho and \popo operations on item $e$ occur at the opposite sides of the deque. Without loss generality, we assume \pusho operation happens at the right side of the deque. 
In this case, we bound the maximum number of items: (i) that are pushed after $t_e^{push}$ from left and are not popped before $t_e^{pop}$; (ii) the items that are already inside the deque at $t_e^{push}$ and are not popped before $t_e^{pop}$. Summing these two terms will provide an upper bound (though not necessarily a tight one) and we obtain $k$. 
Let $size_{j}$ denotes the size of the \subdeque $j$ (number of items inside deque $j$) at time $t_e^{push}$. $pushLeft_{j}$ and $popLeft_{j}$ denotes the number of item that are pushed (and popped) to (from) \subdeque $j$ from left in the time interval $[t_e^{push}, t_e^{pop}]$. We know that for \subdeque $i$, we have $popLeft_{i} = pushLeft_{i} + size_{i}$. And, we have the following three relations: $\forall j \in [1, \width], popLeft_{i} \leq popLeft_{j} + 2\depth,   pushLeft_{i} \geq pushLeft_{i} - 2\depth , size_{i} \geq  size_{j} - 4\depth$. Thus, we obtain $\forall j \in [1, \width], 8\depth \geq pushLeft_{j} + size_{j} - popLeft_{j}$. Summing over all \subdeques $8\depth(\width-1) \geq \sum_{j=0}^{\width-1} pushLeft_{j} + size_{j} - popLeft_{j}$, we obtain an upper bound for the maximum number of items that could be covered by (i) and (ii). 

Finally, we obtain $k$ for a \popo operation by taking the maximum of case one and two:\\ $max(8\depth(\width-1), 4\depth(\width-1)) = 8\depth(\width-1)$.
\end{proof}

\subsection{\cDDc}
\label{sec:cDDccorrectness}

\begin{theorem}
\label{th:cDDc}
\cDDc is linearizable with respect to \koo counter semantics, where $k=(\shift + \depth)(\width-1)$.
\end{theorem}

\begin{proof}

Lemma~\ref{lemma:cor2} states that the size of \substructure can span to at most two consecutive accessible \windows, which implies that the difference between any two \subcounters can be at most $2\depth$ at any point in time. Let $counter_i$ denotes the counter value for \subcounter $i$. One can observe an error at most $(\shift + \depth)(\width-1)$ because $\lVert\sum_{j=1}^{\width} counter_j - (\width) counter_i\rVert \leq (\shift + \depth)(\width-1)$.
\end{proof}

\subsection{\cDDd}
\label{sec:cDDdcorrectness}

\begin{theorem}
\label{th:cDDd}
\cDDd is linearizable with respect to \koo counter semantics, where $k=(2\depth)(\width-1)$.
\end{theorem}

\begin{proof}

Let $counter_i^{decrement}$ (resp. $counter_i^{increment}$) denote \decrement (resp. \increment) counters for \subcounter $i$. We know that $\forall i, j \in [1, \width]: \lVert(counter_i^{increment} - counter_i^{decrement}) - (counter_j^{increment} - counter_j^{decrement})\rVert \leq 2\depth$. Informally, the difference between the sizes of any two \subcounters can be at most $2\depth$ at any point in time. The rest of the proof follows as Theorem~\ref{th:cDDd}.
\end{proof}
\subsection{Lock-freedom}
\wincoupled, \window \shifting is lock-free iff $\shift<\depth$ and obstruction free iff $\shift=\depth$. Take an example, a \getop operation might read an empty \window and \shift it down to a full state, but before it selects a \substructure, a subsequent \putop reads the full \window state and \shifts it up to an empty \window state. It is possible for this process continue forever leading to a system live lock. This is however avoided by setting the \shift parameters to less than \depth.Unlike \wincoupled, \windecoupled is always lock-free. 

Each \substructure is lock-free: An operation can fail on \caeo only if there is another successful operation. A \window \shift can only fail if there is another successful \shift operation preceded by a successful \putop or \getop, ensuring system progress. Thus, all our derived algorithms are lock-free with the exception of possible obstruction freedom as discussed above.
\section{Other Algorithms for Comparison}
\label{sec:otheralgorithms}
To facilitate a detailed study, we implement other given \relaxed data structure algorithms using three extra \relaxation techniques following the same multi \substructure design; \random, \randomc and \roundrobin. These present a combination of characteristics that add value to our evaluation, as shown in Table \ref{table:dsrlx}. Just like our design framework, we use the \width parameter to define the number of \substructures for all the derived algorithms. 

For \random, a \putop or \getop operation selects a \substructure uniformly at random and proceeds to operate on the given \substructure, whereas for \randomc, a \getop operation randomly selects two \substructures, compares their items returning the most correct depending on the data structure semantics \cite{rihani2015brief,Alistarh:2017:PCP:3087801.3087810}. Under \randomc, \putop operations time stamp items marking their time of entry. It is these timestamps that are compared to determine the order among the two items during a \getop operation. Due to the randomised distribution of operations, we expect low or no contention, low or no locality. These three characteristics help us compare and contrast with our optimisations discussed in Section \ref{sec:optimization}. Using the \random design technique we derive a \srandom stack, a \qrandom queue, a \crandom counter and a \drandom deque. Using the \randomc technique we derive a \srandomc stack, a \qrandomc queue and a \crandomc counter. We use similar base algorithms used to derive 2D algorithms discussed in Section \ref{sec:algorithms} for each data structure. \random and \randomc derived algorithms do not provide deterministic \koo relaxation bounds.

Under \roundrobin, a \processor selects and operates on a \substructure in a strict round-robin order following the \processor local counter. The \processor must succeed on the selected \substructure before proceeding to the next. This implies that in case of contention, a failing \processor must retry on the same \substructure until the given \processor succeeds. Due to retries on contended \substructures by the contending \processors, we expect more contention to be generated from the retries. The \processor always selects a different \substructure after a successful operation, hence we expect low or no locality. However, memory access using round robin scheduling can take advantage of hardware prefecthing, a good characteristic to compare and contrast with our locality optimisation. Using \roundrobin technique, we derive a \srobin stack, a \qrobin queue, a \crobin counter and a \drobin. \roundrobin provides deterministic \koo \relaxation bounds, we demonstrate this using \srobin whose bound is given by Theorem \ref{th:srobin}.

\subsection{\srobin Correctness}
\label{app:srobincorrectness}

\begin{theorem}
\label{th:srobin}
\srobin is linearizable with respect to \koo stack semantics, where $k=(2P-1)(\width-1)$. Where $P$ is number of \processors and \width is the number of \substaks
\end{theorem}

\begin{proof}

Consider the linearization points of \pusho and \popo operations that respectively insert and remove the
item $e$ into and from a \substak (let \substak $0$). 
Let $t_e^{push}$ and $t_e^{pop}$ denote 
these points, respectively. Now, we bound the maximum number of items, that are
pushed after $t_e^{push}$ and are not popped before $t_e^{pop}$, to obtain $k$.
We denote the number items that are pushed to (popped from) \substak $i$ by 
thread $j$ in the time interval $[t_e^{push}, t_e^{pop}]$, with ${push}_{i}^{j}$ 
(${pop}_{i}^{j}$).

Observe that each thread applies its operations in round robin fashion without 
skipping any \substak. If the previous successful \popo had occurred at
\substak $i$, the next \popo occurs at \substak $i+1 (mod\text{ }\width)$. The 
same applies for the push operations. 

Without loss of generality, assume that thread $0$ has inserted item $e$ to \substak $0$.
This implies that $\forall i, \width-1 \geq i>0, \ {push}_{0}^{0} \geq {push}_{i}^{0}$. 
Now, take another thread $j$, we have 
$\forall i: \width-1 \geq i>0, {push}_{0}^{j} \geq {push}_{i}^{j} - 1$. Informally, another 
thread can increase the number of items on any other \substak by at most one more compared to the number of items that pushes on \substak $0$.

For the pop operations, we have the same relation for all threads: $\forall i, 
\width \geq i > 0, \ {pop}_{0}^{j} \geq {pop}_{i}^{j} + 1$. Informally, a 
thread can pop at most 1 item less from any other \substak compared to the number that it pops from \substak $0$.
As the interval $[t_e^{push}, t_e^{pop}]$ starts with the push and ends with 
the pop of item $e$ at \substak $0$, we have 
$\sum_{j=0}^{P-1} {push}_{0}^{j} = \sum_{j=0}^{P-1} {pop}_{0}^{j} = Y$.

Summing over all threads and \substaks other than \substak $0$, we get at most 
$(Y+P-1)(\width-1)$ \pusho operations in the interval $[t_e^{push}, t_e^{pop}]$.
Summing over all threads and \substaks other than \substak $0$, we get at least
$(Y-P)(\width-1)$ \popo operations. Which leads to the theorem: 
$k \leq ((Y+P-1)-(Y-P))(\width-1) = (2P-1)(\width-1)$
\end{proof}
\section{Experimental Evaluation}
\label{sec:evaluation}

\begin{table}\centering \small
  \begin{tabular}{l l l l}
    \hline
    Algorithm & \width \\ \hline
    \random (Stack,Counter, Queue \& Deque) & $3P$  \\
    \randomc (Stack,Counter \& Queue) & $3P$  \\
    \srobin & $(k/(2P-1))+1$  \\
    \crobin & $(k/(2P-1))+1$ \\
    \qrobin & $(k/(P-1))+1$ \\
    \drobin & $(k/(2P-1))+1$ \\ 
    \kstack & $k+1$ \\
    \qsegment & $k+1$ \\
    \lru & $k+1$ \\
    \sDDd & $(k/(3\depth)) + 1$ \\
    \sDDc & $(k/(2\shift + \depth)) + 1$ \\
    \qDDd & $(k/\depth) + 1$ \\
    \dDDd & $(k/3\depth) + 1$ \\
    \cDDc & $(k/(\shift + \depth)) + 1$ \\
    \cDDd & $(k/2\depth) + 1$ \\
  \end{tabular}
  \caption{Execution parameter configuration for the different algorithms}
  \label{table:dsrlx}
\end{table}

We experimentally evaluate the performance of our derived algorithms, in comparison to \koo \relaxed algorithms available in the literature, and other state of the art data structure algorithms. \koo \relaxed algorithms include; Least recently used queue (\lru) \cite{Haas:2013:DQS:2482767.2482789}, Segmented queue (\qsegment) and \kstack \cite{afek2010quasi,henzinger2013quantitative},  other algorithms include; MS-queue (\msqueue) \cite{Michael:1996:SFP:248052.248106}, Wait free queue (\wfqueue) \cite{Yang:2016:WQF:2851141.2851168}, general doubly linked list Deque (\dsundell) \cite{DBLP:journals/jpdc/SundellT08}, Maged deque (\dmaged) \cite{DBLP:conf/europar/Michael03}, Time stamped stack (\tsstack) \cite{Dodds:2015:SCT:2775051.2676963} and Elimination back-off stack (\elimination) \cite{hendler2010scalable}. \width will be generally used to refer to number of \substructures for all algorithms using multiple \substructures.

To facilitate a uniform comparison, we implemented all the evaluated algorithms using the same development tools and environment. The source code is publicly available on this link https://github.com/dcs-chalmers/Semantic-relaxation-2D-design-framework. 

\subsection{System Description}
Experiments are run on two x86-64 machines: (i) Intel Xeon E5-2687W v2 dual socket machine (\multisocket) and (ii) Intel Xeon Phi 7290 single socket machine (\singlesocket). \multisocket has 8-core Intel Xeon processors per socket, each running at 3.4GHz, with three cache levels: core private L1 cache = 32 KB, core private L2 cache = 256 KB and a socket intra shared L3 cache = 25.6 MB. \singlesocket has 72-core processor running at 1.5GHz, core private L1 cache = 32 KB, L2 cache = 1024KB shared by two cores (tile). \multisocket and \singlesocket run on Ubuntu 16.04.2 LTS and CentOS Linux 7 Core Operating systems receptively. The \multisocket machine is used to evaluate inter-socket execution behaviour, whereas \singlesocket is used to evaluate intra-socket with high number of \processors. \Processors are pined one per core, for both machines excluding hyper-threading. Inter-socket execution is evaluated through pinning the \processors one per core per socket in round robin fashion. \Processors randomly select between \putop or \getop per operation with a given probability (operation rate). Memory is managed using the ASCYLIB framework SSMEM \cite{David:2015:ACS:2786763.2694359}.

Our main goal is to achieve scalability under high operation rate. To evaluate this, we simulate high operation rate by excluding work between operations. To reduce the effect of $NULL$ returns \footnote{Usually $NULL$ returns are very fast because they do not update the data structure state, if not minimised, they can give misleading performance results} for \getop operations, all algorithms are initialised with $2^{17}$ items. Each experiment is then run for five seconds obtaining an average of five repeats. Throughput is measured in terms of operations per second, whereas the \relaxation behaviour (\accuracy) is measured in terms of the error distance from the exact data structure sequential semantics \cite{henzinger2013quantitative}. The higher the error distance, the lower the \accuracy.

\subsection{Measuring \accuracy} We adopt a similar methodology used in the literature \cite{alistarh2015spraylist,rihani2015brief}. 
A sequential linked-list (doubly-linked list for deque) is run alongside the data structure being measured. Items on the data structure are duplicated on the linked-list and can be identified by their unique labels. For each operation \putop or \getop, a simultaneous insert or delete is performed on the linked-list respectively, following the exact semantics of the given data structure. A global lock is carefully placed at the data structure linearization points, locking both the linked-list and the data structure simultaneously. The lock allows only one \processor to update both the data structure and the linked-list in isolation. 

A given \processor has to acquire the lock before it tries to linearize on any given \substructure. Note that, \window search is independent of the lock. \putop operations happen at the head or tail of the list for LIFO or FIFO measurements respectively (left or right for deque). A \getop operation searches for the given item on the linked list, deletes it, and returns its distance from the respective access point, head or tail (error distance). For counter measurements, we replace the linked-list with a single fetch and add (\faa) counter. Both counters are updated in isolation using a global lock like explained above. The error distance is calculated from the difference between the two counter values.

Experiment results are then plotted using logarithmic scales, throughput (solid lines) and error distance (dotted lines) sharing the x-axis.
\subsection{Dimension Tunability}
\begin{figure*}
\begin{minipage}{1\textwidth}
\centering
\includegraphics[scale=0.5,trim={0 1.0cm 0 0},clip]{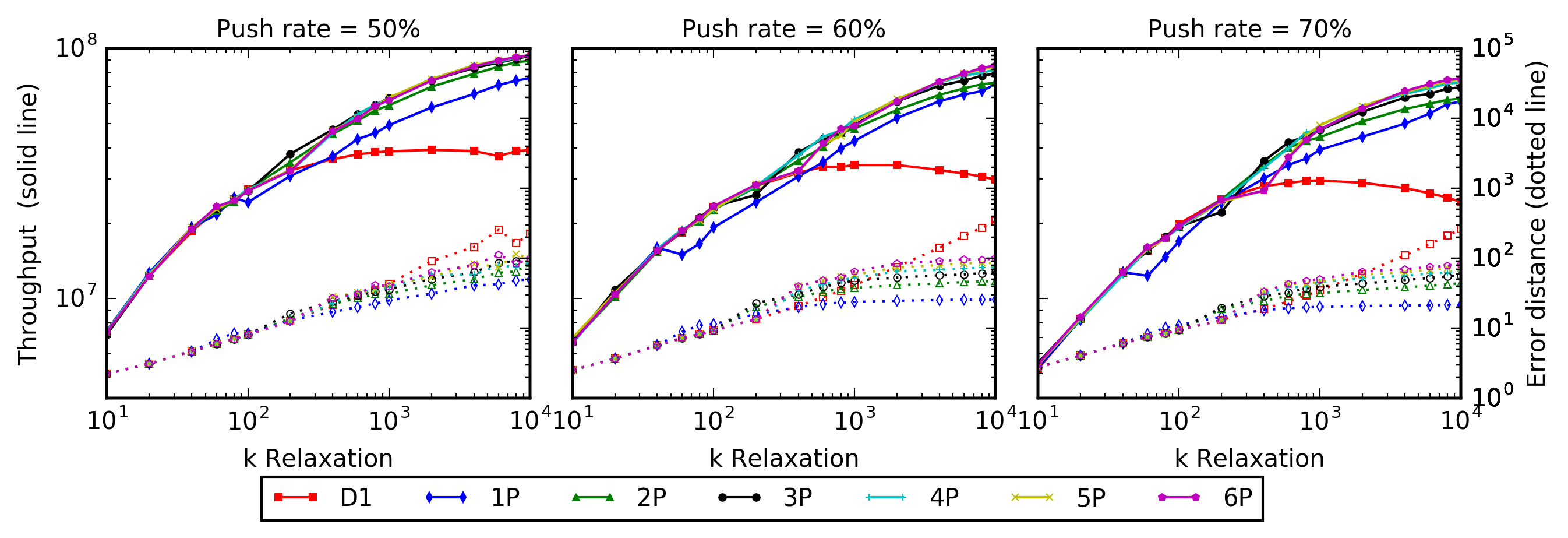}
\subcaption{\singlesocket}
\label{fig:wincoupleduma}
\end{minipage}
\begin{minipage}{1\textwidth}
\centering
\includegraphics[scale=0.5]{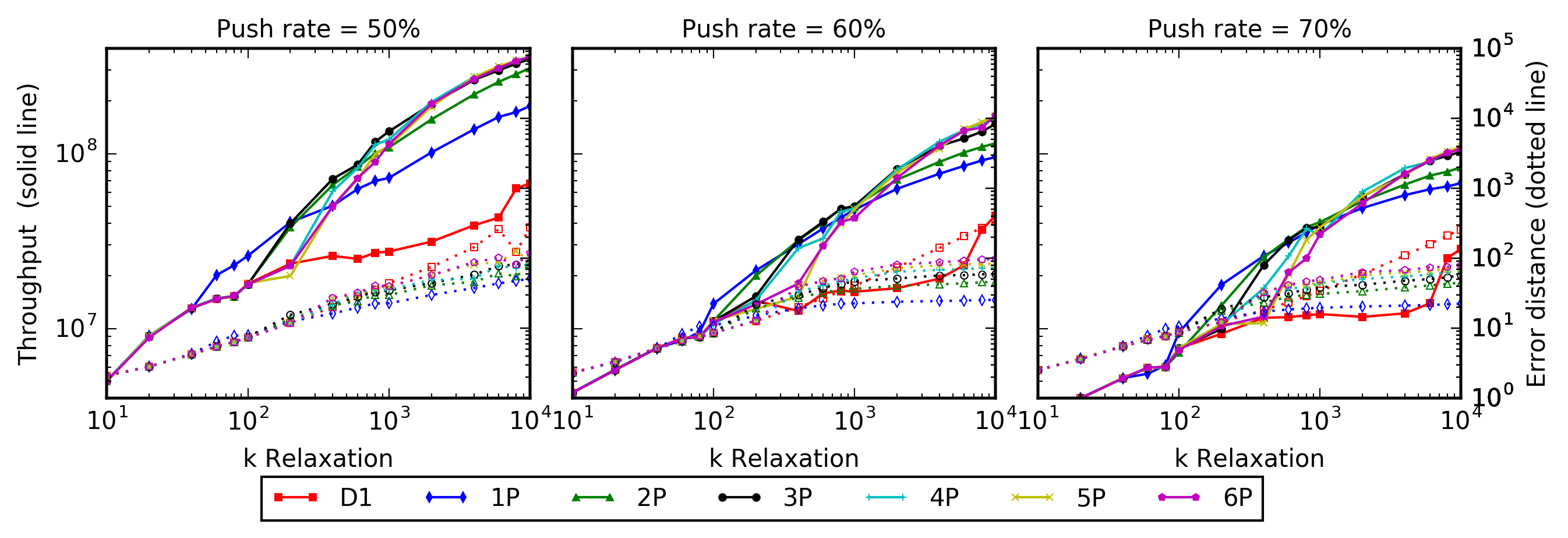}
\subcaption{\multisocket}
\label{fig:wincouplednuma}
\end{minipage}
\caption{\sDDc exemplifying the \wincoupled performance with different \width configurations ($P=16$).}
\label{fig:wincoupled}
\vspace{10pt}
\begin{minipage}{1\textwidth}
\centering
\includegraphics[scale=0.5,trim={0 1.0cm 0 0},clip]{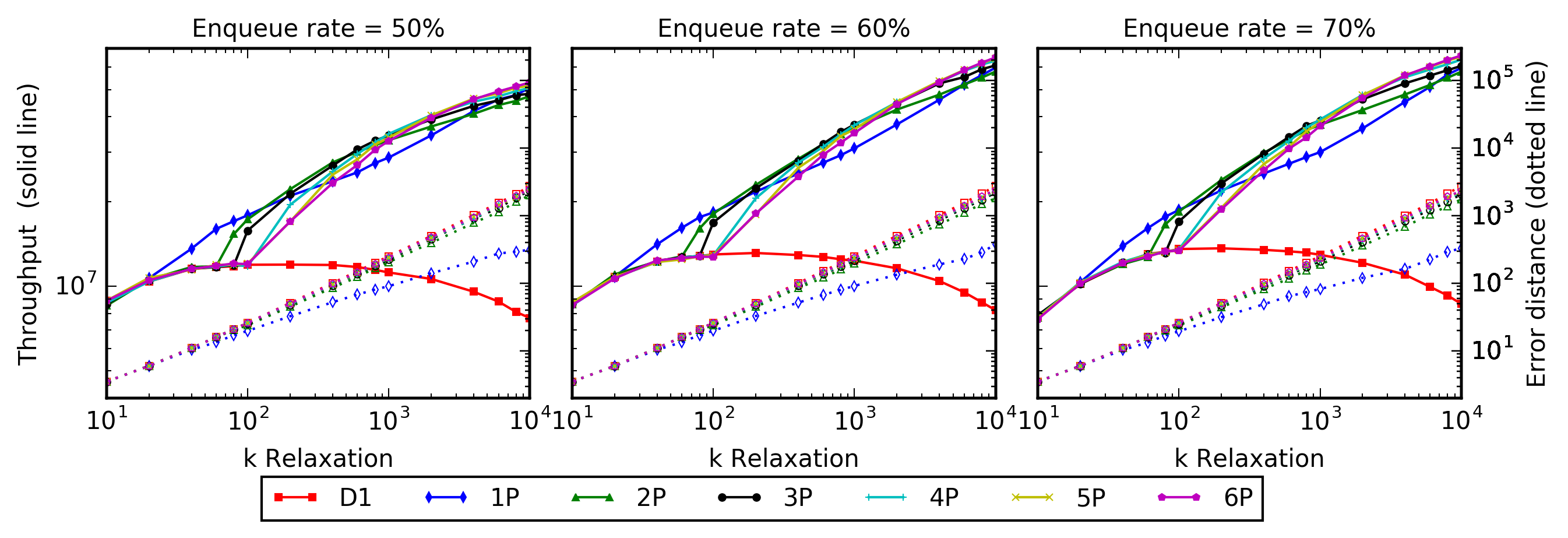}
\subcaption{\singlesocket}
\label{fig:windecoupleduma}
\end{minipage}
\begin{minipage}{1\textwidth}
\centering
\includegraphics[scale=0.5]{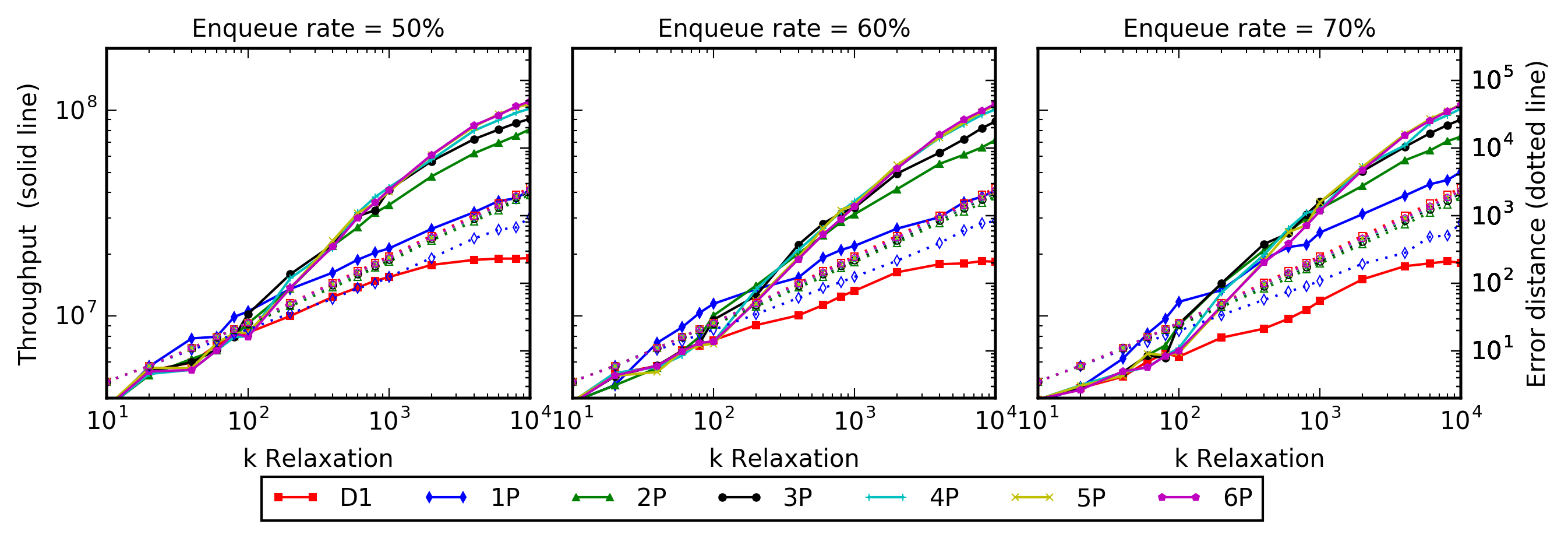}
\subcaption{\multisocket}
\label{fig:windecouplednuma}
\end{minipage}
\caption{\qDDd exemplifying the \windecoupled performance with different \width configurations ($P=16$).}
\label{fig:windecoupled}
\end{figure*}

Our design framework is tunable, giving designers the ability to manage performance optimizations for different execution environments and workloads,  within a given tight \relaxation bound ($k$). To evaluate this, we experiment with different \width and \depth parameter configurations, as shown in Figures \ref{fig:wincoupled} and \ref{fig:windecoupled} for \sDDc (representative of \wincoupled) and \qDDd (representative of \windecoupled) respectively. Curve (D1) depicts the case for fixed $\depth = 1$ which also represents a case of \relaxing in one dimension (increasing the number of \substructures). The other curves (1P,2P,3P,4P,5P and 6P) depict execution in two dimensions with a set maximum $\width \in \{1P,2P,3P,4P,5P, 6P\}$, where $P$ is the number of \processors. As \relaxation increases, we increase \width until the maximum \width, then switch to increasing \depth within the specified \relaxation bound. For simplicity, \width is described as a multiple of the number of \processors. It should however be noted that \width can be configured to be independent of the number of \processors.

D1 presents the lowest throughput as $k$ increases. This is attributed to the increasing \width proportional to $k$, leading to increased \hopos and lack of locality exploitation. On the other hand, we observe improved throughput performance for two dimensional executions. For all measured $k$, we observe that there is no consistent optimal \width configuration. This implies that, an optimal configuration is dependant on other factors, including; $k$ \relaxation, type of workload, plus \accuracy vs throughput trade-off. There are also notable differences between \singlesocket and \multisocket results. This calls for a multi-objective optimisation model, which is beyond the scope of this article. 

With respect to the evaluated cases, we observe that $\width = 3P$ provides a fair balance between \accuracy and throughput performance especially for \sDDc as shown in Figure \ref{fig:wincoupled}. However for smaller $k$, there are varying high throughput points. Since we do not have an optimisation model, we empirically obtain the high throughput \width configurations for different executions.

However for \multisocket, we note that smaller \width ($1P$) can achieve higher throughput performance. This is attributed to the high inter socket communication cost, having smaller \width allows for the exploitation of locality through \relaxing more in the \depth dimension. Exploiting locality reduces the communication between sockets, in turn, avoiding the inter socket communication cost.
\subsection{Scaling With Relaxation}
\label{sec:rlxevaluation}
\begin{figure*}
\begin{minipage}{1\textwidth}
\centering
\includegraphics[scale=0.5,trim={0 1.0cm 0 0},clip]{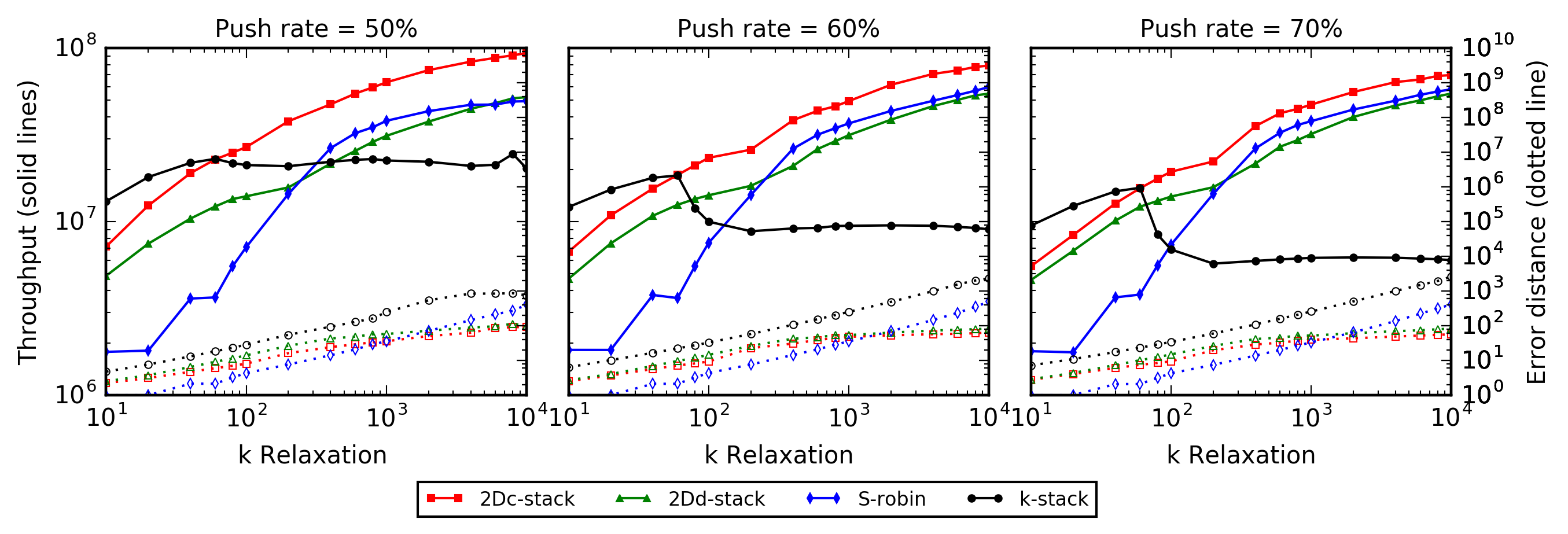}
\subcaption{\singlesocket}
\label{fig:sumarelaxationp16}
\end{minipage}
\begin{minipage}{1\textwidth}
\centering
\includegraphics[scale=0.5]{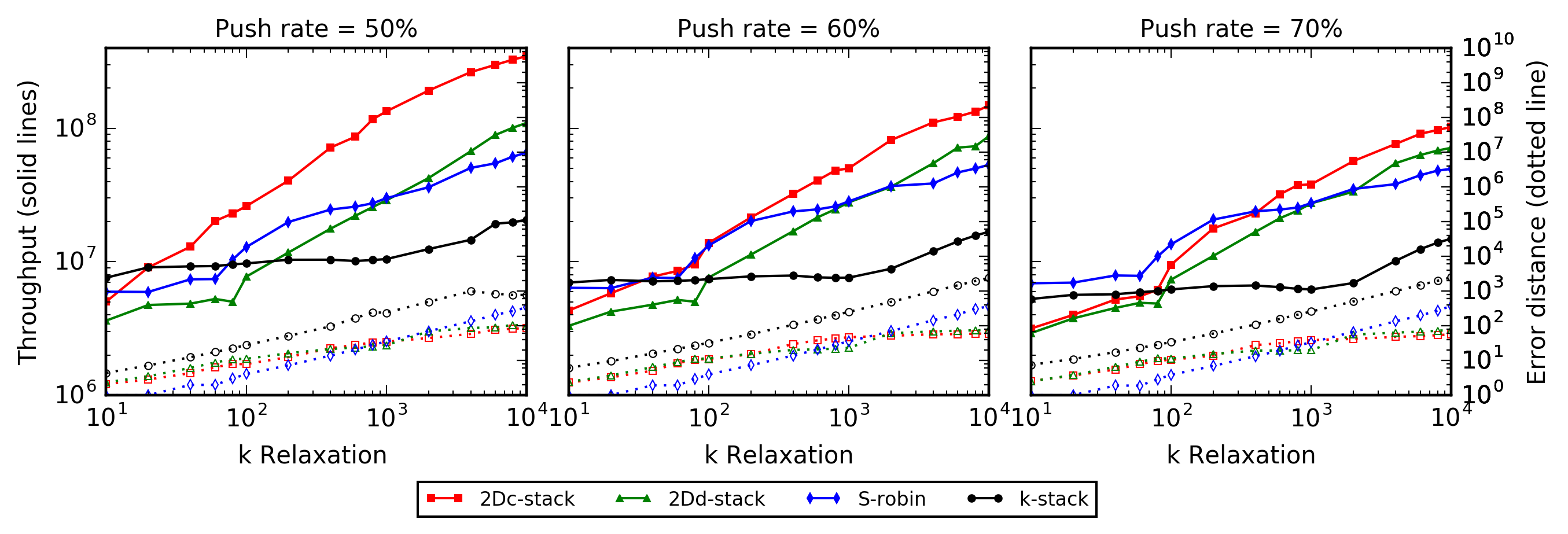}
\subcaption{\multisocket}
\label{fig:snumarelaxationp16}
\end{minipage}
\caption{Stack throughput and observed \accuracy as $k$ bound relaxation increases ($P=16$).}
\label{fig:srelaxationp16}
\vspace{10pt}
\begin{minipage}{1\textwidth}
\centering
\includegraphics[scale=0.5,trim={0 1.0cm 0 0},clip]{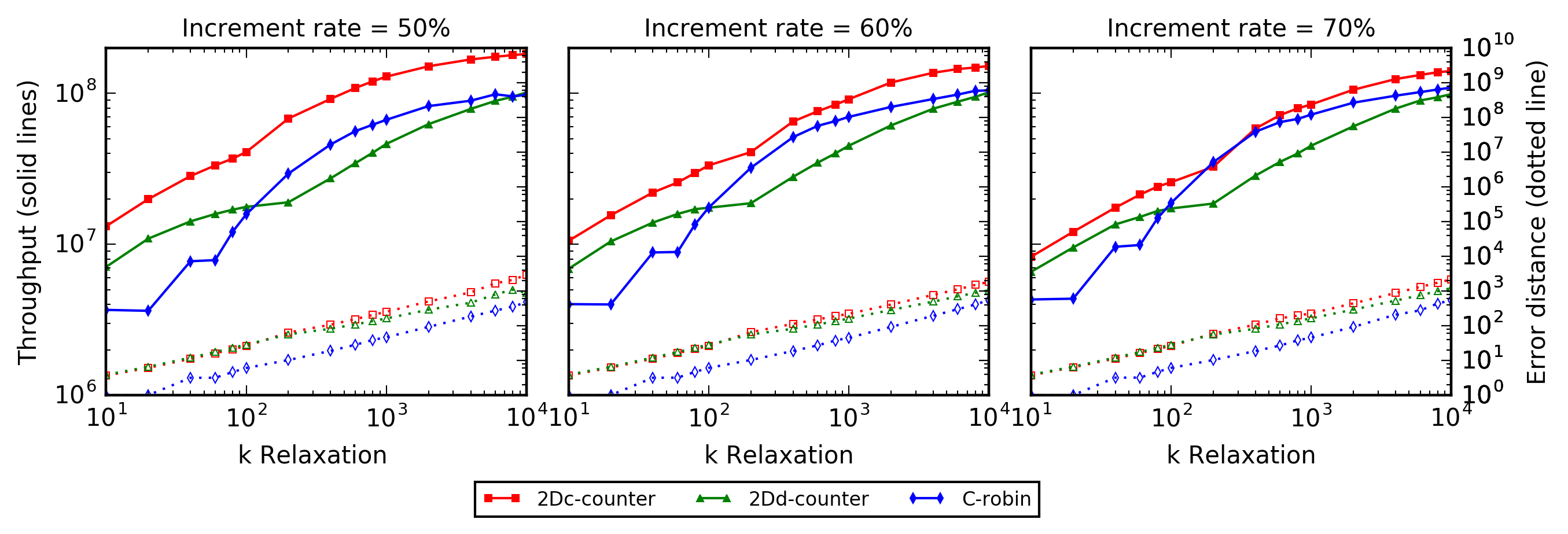}
\subcaption{\singlesocket}
\label{fig:cumarelaxationp16}
\end{minipage}
\begin{minipage}{1\textwidth}
\centering
\includegraphics[scale=0.5]{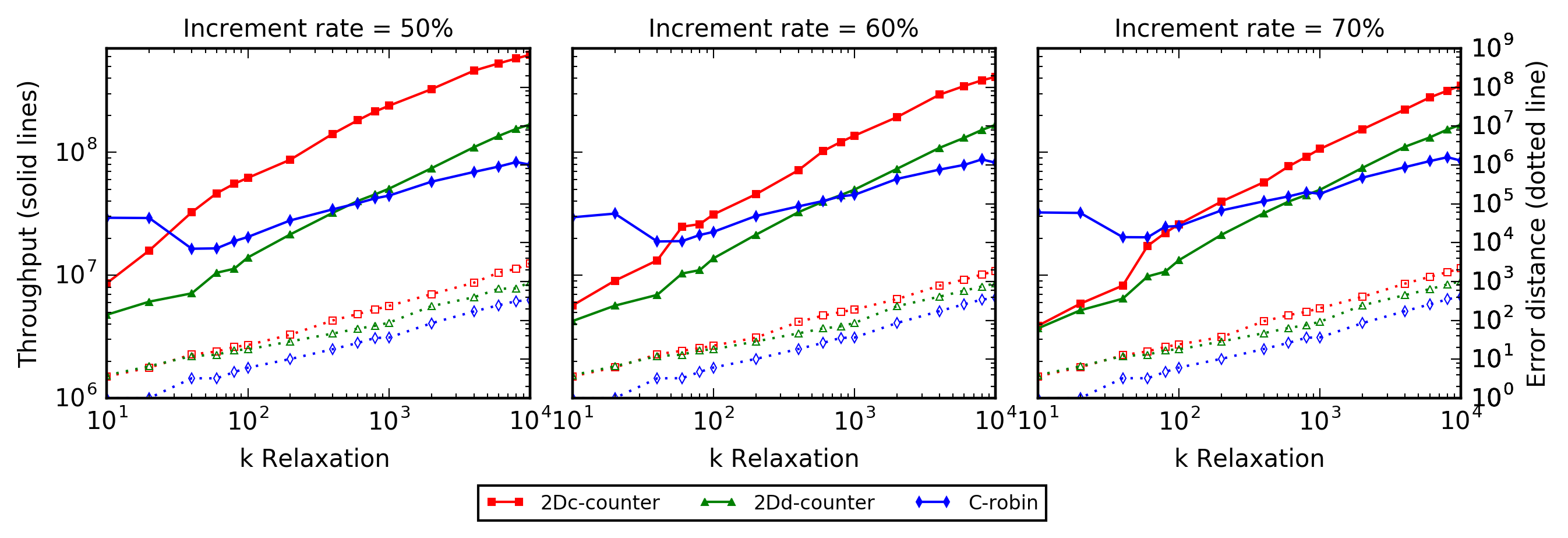}
\subcaption{\multisocket}
\label{fig:cnumarelaxationp16}
\end{minipage}
\caption{Counter throughput and observed \accuracy as $k$ bound relaxation increases ($P=16$).}
\label{fig:crelaxationp16}
\end{figure*}

\begin{figure*}
\begin{minipage}{1\textwidth}
\centering
\includegraphics[scale=0.5,trim={0 1.0cm 0 0},clip]{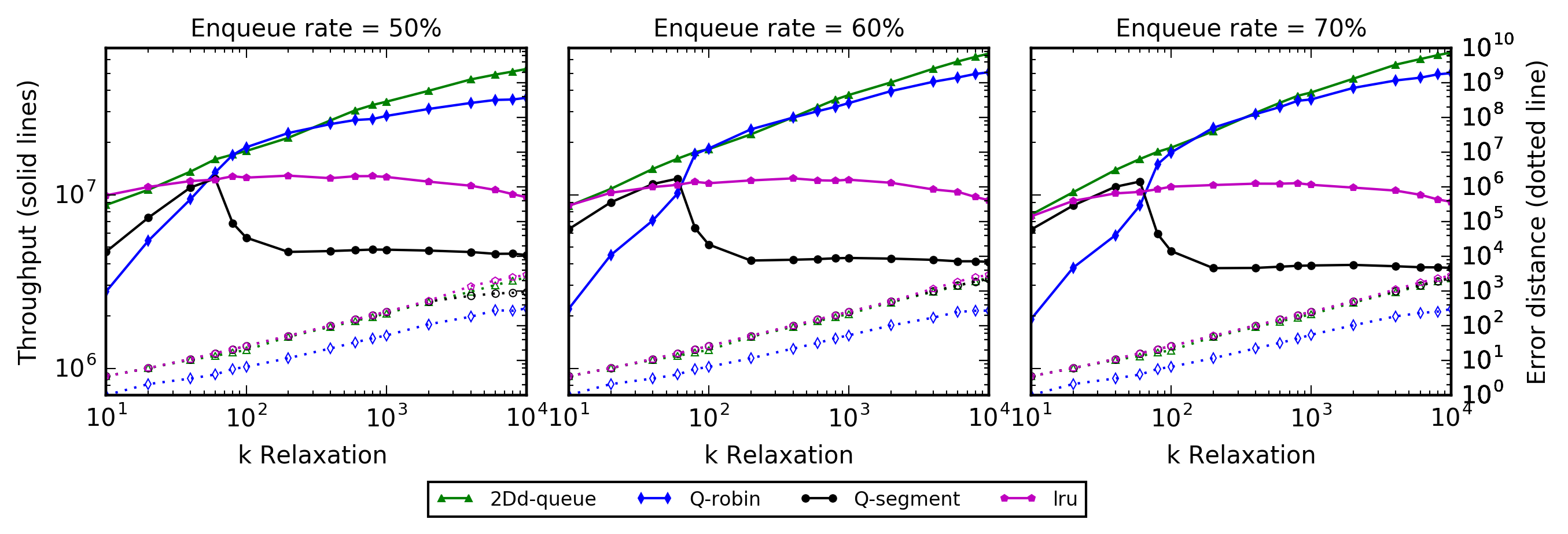}
\subcaption{\singlesocket}
\label{fig:qumarelaxationp16}
\end{minipage}
\begin{minipage}{1\textwidth}
\centering
\includegraphics[scale=0.5]{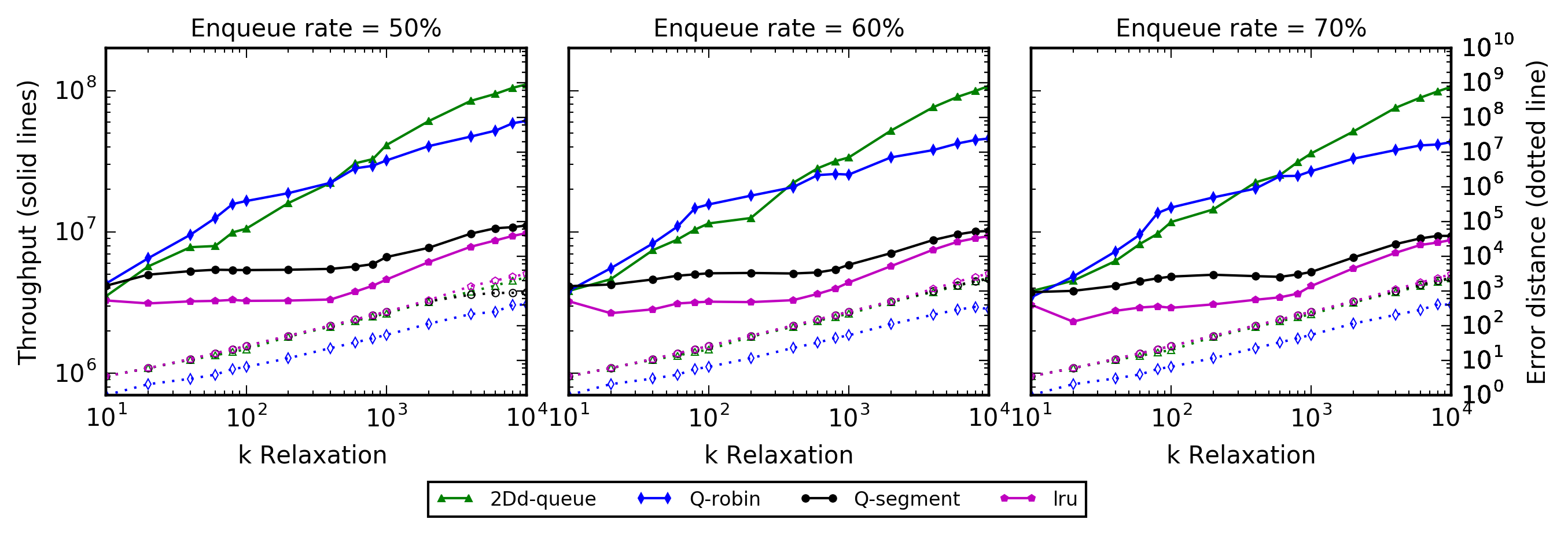}
\subcaption{\multisocket}
\label{fig:qnumarelaxationp16}
\end{minipage}
\caption{Queue throughput and observed \accuracy as $k$ bound relaxation increases ($P=16$).}
\label{fig:qrelaxationp16}
\vspace{10pt}
\begin{minipage}{1\textwidth}
\centering
\includegraphics[scale=0.5,trim={0 1.0cm 0 0},clip]{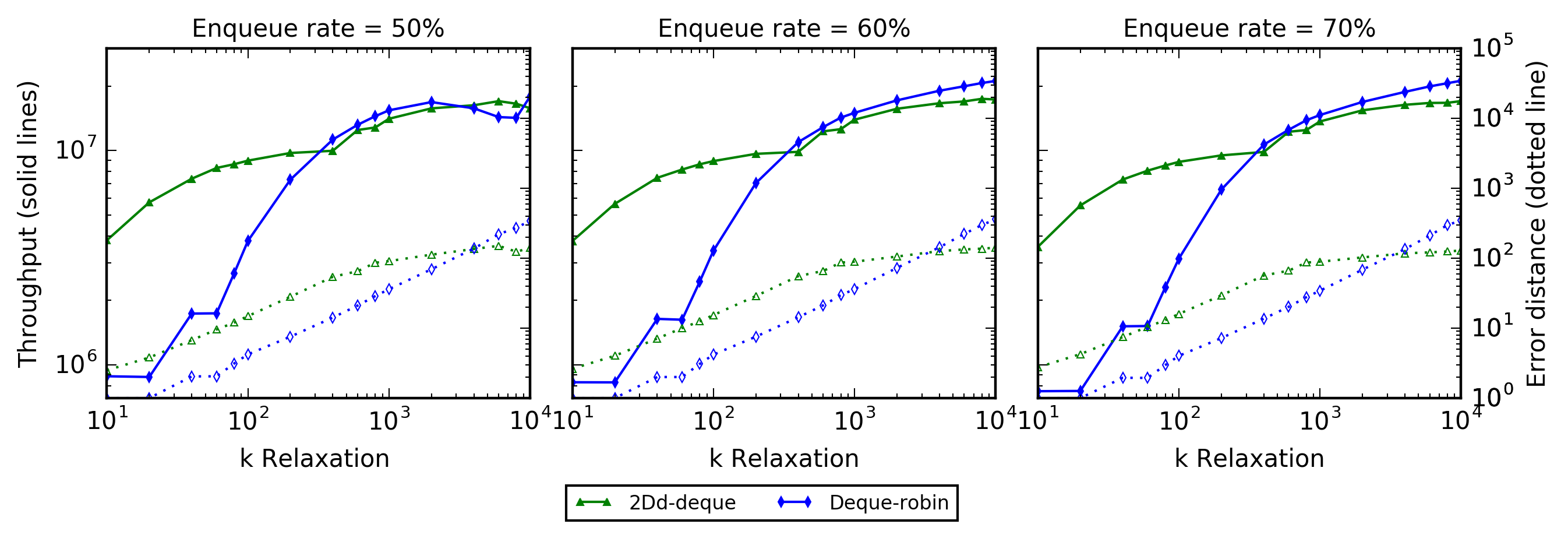}
\subcaption{\singlesocket}
\label{fig:dumarelaxationp16}
\end{minipage}
\begin{minipage}{1\textwidth}
\centering
\includegraphics[scale=0.5]{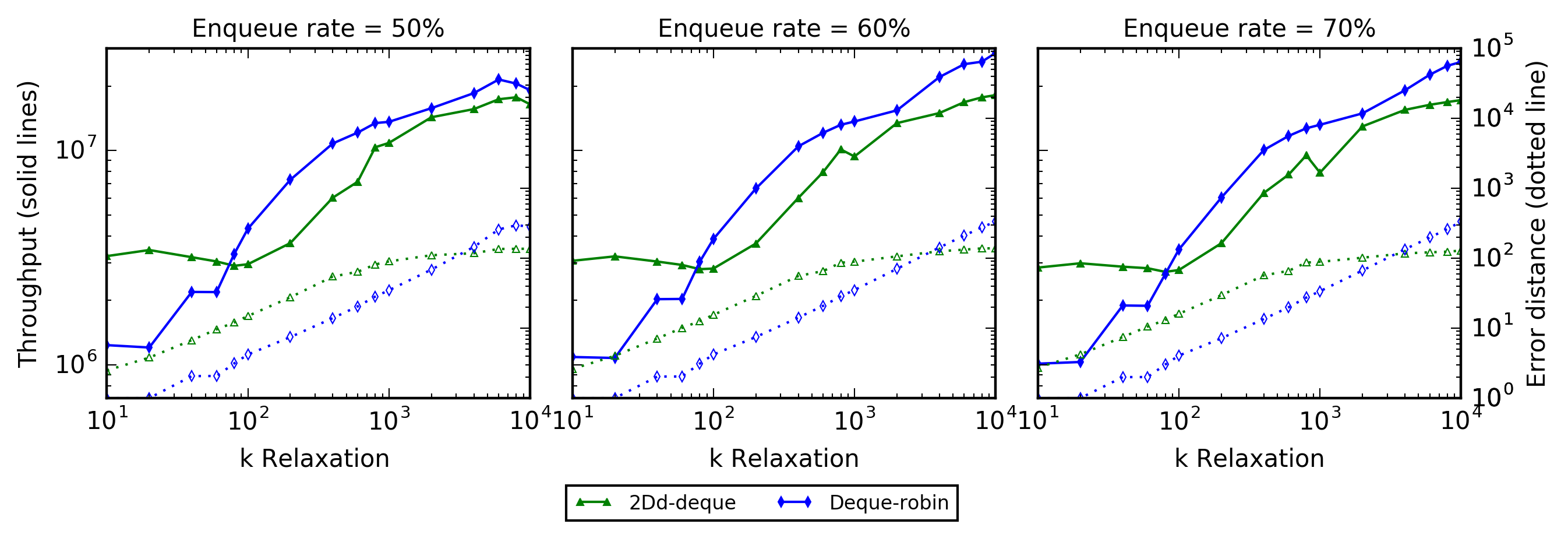}
\subcaption{\multisocket}
\label{fig:dnumarelaxationp16}
\end{minipage}
\caption{Deque throughput and observed \accuracy as $k$ bound relaxation increases ($P=16$).}
\label{fig:drelaxationp16}
\end{figure*}

In order to evaluate monotonicity with increasing \relaxation bound ($k$), we fix the number of \processors to $16$. This is to match the number of cores available on \multisocket without hyper-threading. Results are presented in Figures \ref{fig:srelaxationp16}, \ref{fig:qrelaxationp16}, \ref{fig:crelaxationp16} and \ref{fig:drelaxationp16} for stack, queue, counter and deque respectively. 

First, we observe the difference between \wincoupled and \windecoupled for \sDDc and \sDDd respectively in Figure \ref{fig:srelaxationp16}. \sDDc consistently outperforms \sDDd due to the reduced \window \shifting updates. With \sDDc, a given \processor can locally operate on the same \substak longer since operation counts cancel out each other leaving the \substak in a valid state. The longer a given \substak stays valid, the higher the chances of exploiting locality. This advantage is more evident with symmetric workloads (50\% push rate). As the workload becomes more asymmetric (70\% push-rate), less \pusho counts are cancelled out by \popo counts. This implies that, the \window gets full more frequent leading to increased \window \shifts. With 100\% asymmetric workloads, \sDDc and \sDDd present similar execution behaviour. The same is observed for \cDDc and \cDDd in Figure \ref{fig:crelaxationp16}.

All multi \substructure based algorithms increase their \width (number of \substructures) as $k$ increases to reduce contention and allow for increased disjoint \access as shown in Table \ref{table:dsrlx}. However, for \kstack, \qsegment, and \lru, \hopos increase as \width increases, this explains their observed low throughput gain. \srobin, \qrobin and \crobin are not affected by \hopos. However, for smaller $k$ values, they suffer from high contention arising from contending \processors retrying on the same \substructure until they succeed. As contention vanishes with high $k$ values, throughput gain saturates due to lack of locality. \roundrobin algorithms take advantage of the hardware prefetching available on both \singlesocket and \multisocket machines to reduce on the downside effect of lack or locality. This explains the observed throughput gain as \width increases.   

$2D$ algorithms maintain throughput gain through limiting \width to a size beneficial to reducing contention and switch to adjusting the \depth to reduce \hopos. For our evaluation, $2D$ algorithms' \width is configured as shown in Table \ref{table:dsrlx}. Once the algorithm attains the configured \width, it switches to increasing \depth as $k$ increases. The \depth parameter allows $2D$ algorithms to maintain throughput gain (monotonicity) through exploiting locality while reducing latency. This is observed for both \singlesocket and \multisocket machines.

In terms of \accuracy, we observe an almost linear decrease in \accuracy as $k$ increases for all algorithms. However, for \sDDc we observe a reduced rate of \accuracy loss when the algorithm switches to increasing \depth.

\subsection{Scaling With \Processors}

\begin{figure*}
\begin{minipage}{1\textwidth}
\centering
\includegraphics[scale=0.5,trim={0 1.0cm 0 0},clip]{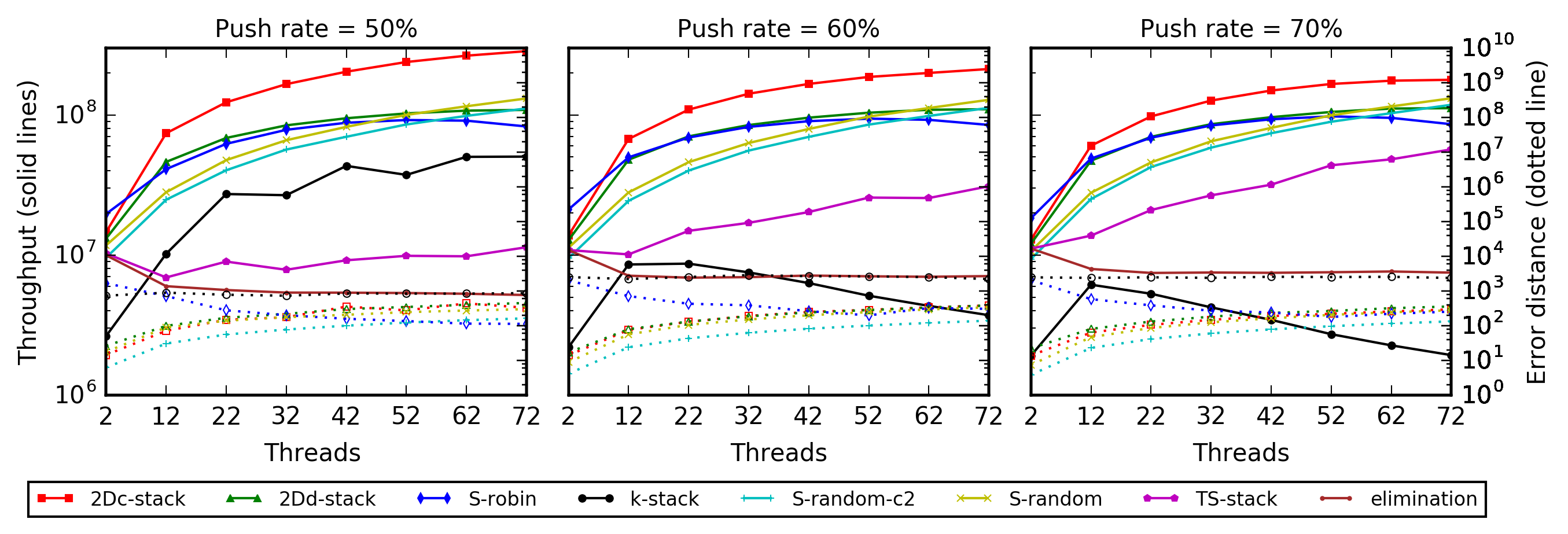}
\subcaption{\singlesocket}
\label{fig:sumaconcurrencyk10000}
\end{minipage}
\begin{minipage}{1\textwidth}
\centering
\includegraphics[scale=0.5]{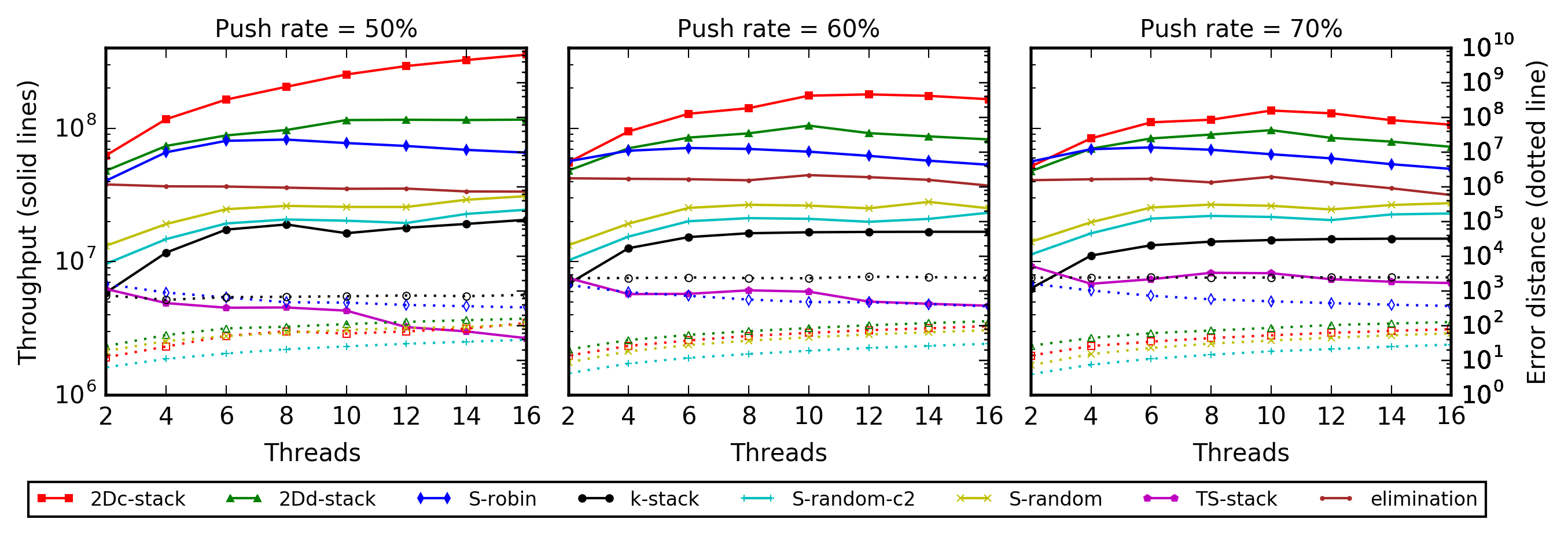}
\subcaption{\multisocket}
\label{fig:snumaconcurrencyk10000}
\end{minipage}
\caption{Stack throughput and observed \accuracy as the number of \processors increases ($k=10^4$).}
\label{fig:sconcurrencyk10000}
\vspace{10pt}
\begin{minipage}{1\textwidth}
\centering
\includegraphics[scale=0.5,trim={0 1.0cm 0 0},clip]{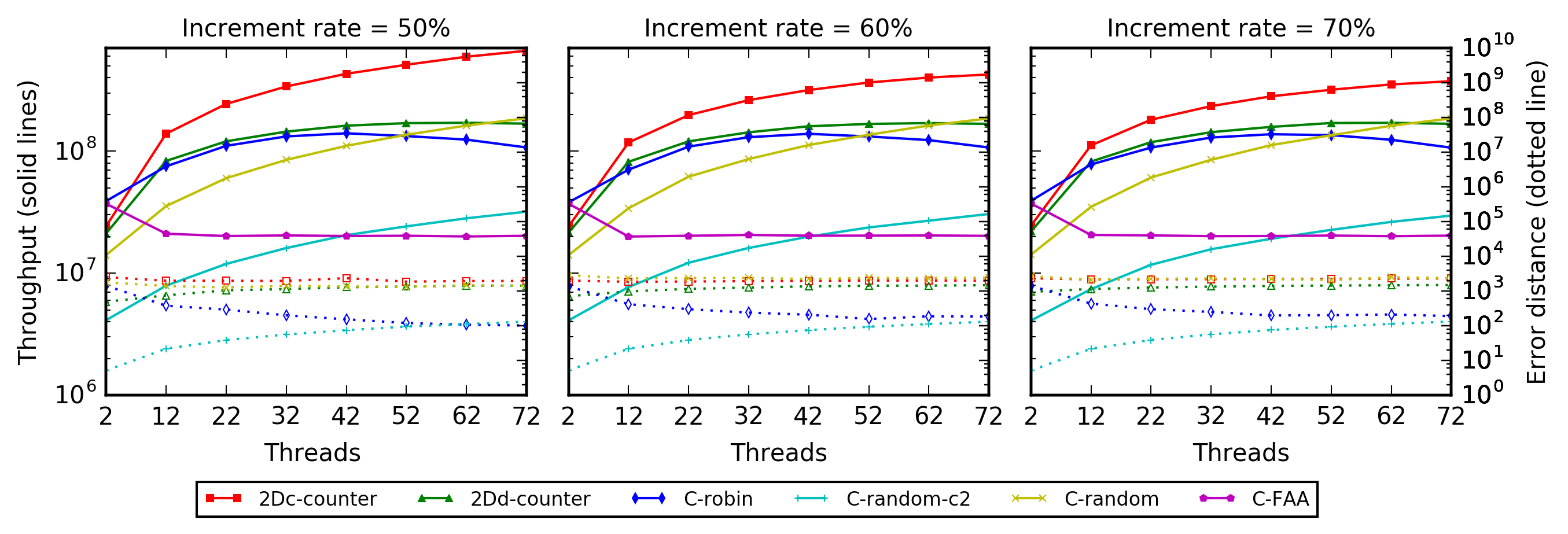}
\subcaption{\singlesocket}
\label{fig:cumaconcurrencyk10000}
\end{minipage}
\begin{minipage}{1\textwidth}
\centering
\includegraphics[scale=0.5]{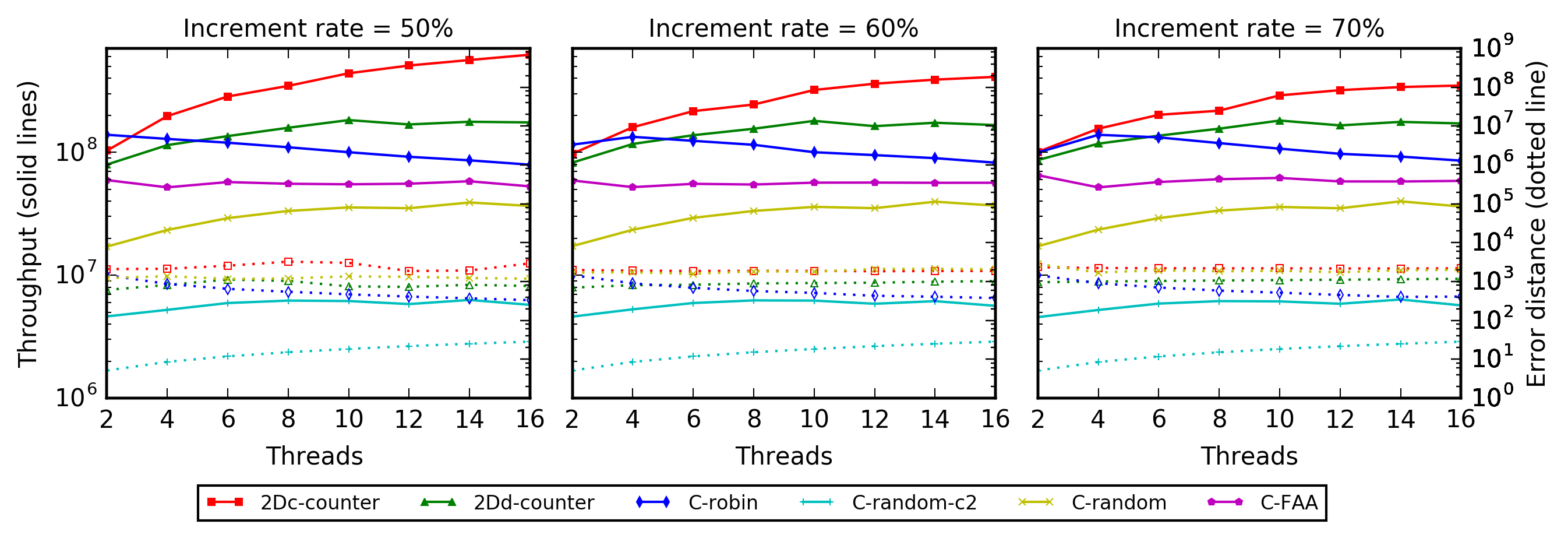}
\subcaption{\multisocket}
\label{fig:cnumaconcurrencyk10000}
\end{minipage}
\caption{Counter throughput and observed \accuracy as the number \processor increases ($k=10^4$).}
\label{fig:cconcurrencyk10000}
\end{figure*}

\begin{figure*}
\begin{minipage}{1\textwidth}
\centering
\includegraphics[scale=0.5,trim={0 1.0cm 0 0},clip]{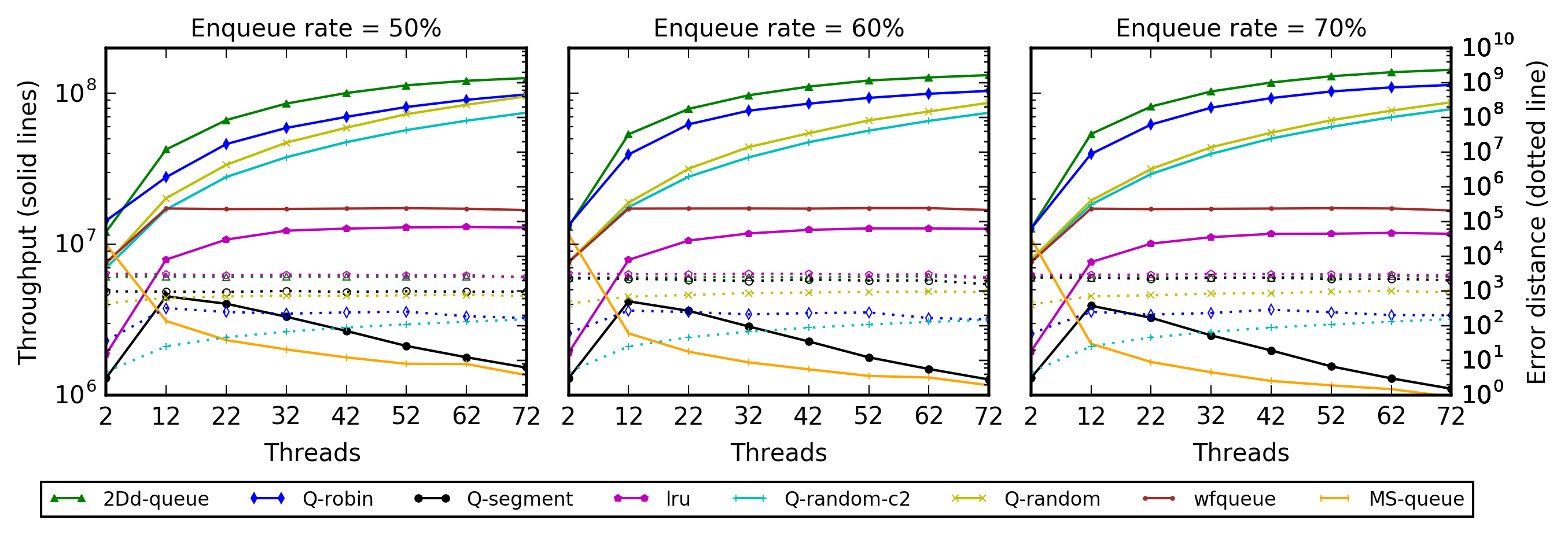}
\subcaption{\singlesocket}
\label{fig:qumaconcurrencyk10000}
\end{minipage}
\begin{minipage}{1\textwidth}
\centering
\includegraphics[scale=0.5]{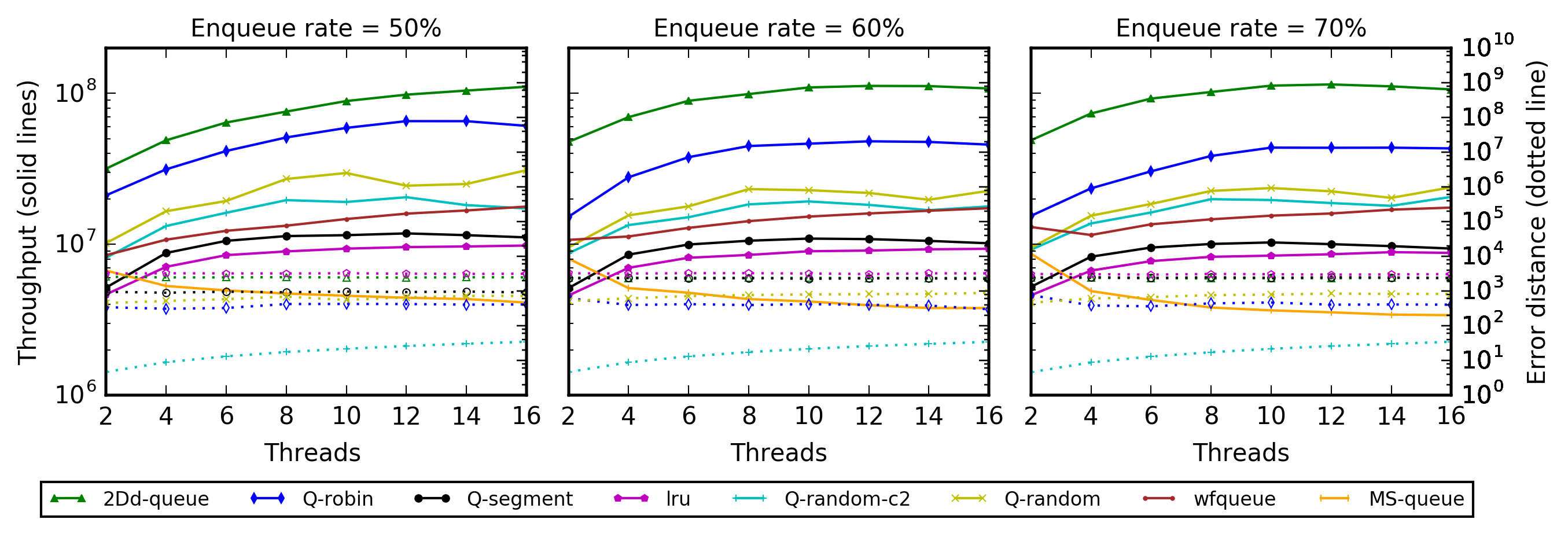}
\subcaption{\multisocket}
\label{fig:qnumaconcurrencyk10000}
\end{minipage}
\caption{Queue throughput and observed \accuracy as the number of \processors increases ($k=10^4$).}
\label{fig:qconcurrencyk10000}
\vspace{10pt}
\begin{minipage}{1\textwidth}
\centering
\includegraphics[scale=0.5,trim={0 1.0cm 0 0},clip]{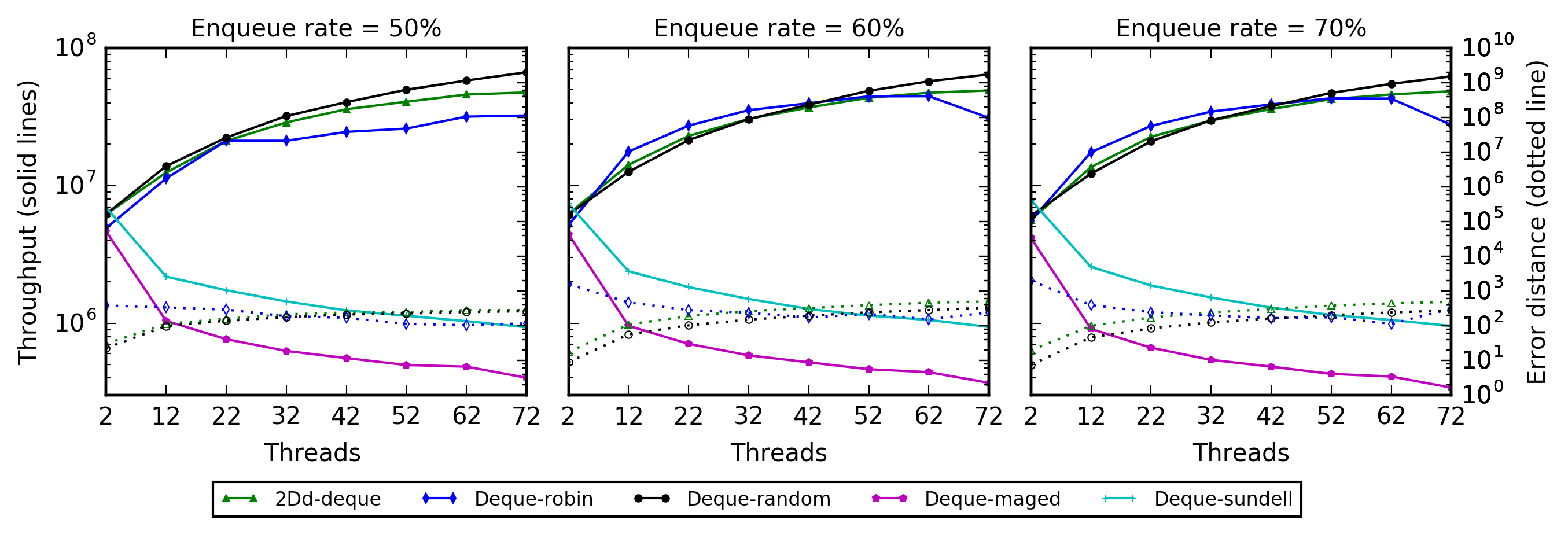}
\subcaption{\singlesocket}
\label{fig:dumaconcurrencyk10000}
\end{minipage}
\begin{minipage}{1\textwidth}
\centering
\includegraphics[scale=0.5]{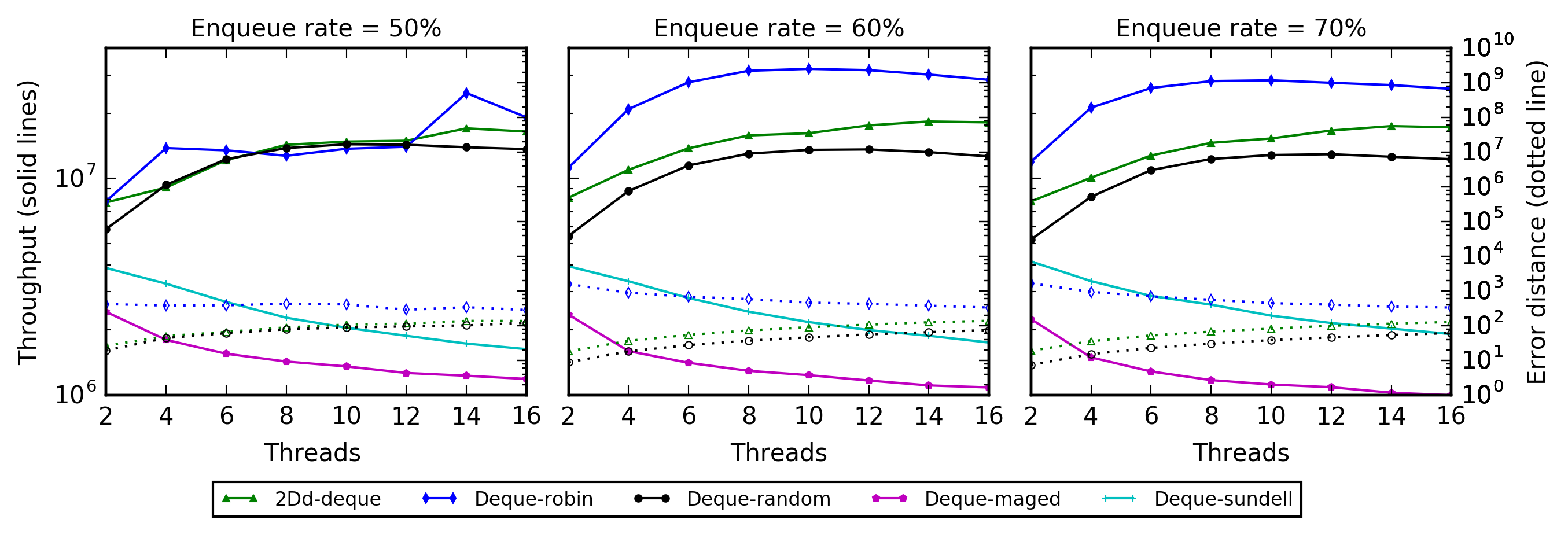}
\subcaption{\multisocket}
\label{fig:dnumaconcurrencyk10000}
\end{minipage}
\caption{Deque throughput and observed \accuracy as the number of \processors increases ($k=10^4$).}
\label{fig:dconcurrencyk10000}
\end{figure*}

To evaluate the scalability of our design as the number of \processors increases, we fix the \relaxation bound to ($k=10^4$) and vary the number of \processors as shown in Figures \ref{fig:sconcurrencyk10000}, \ref{fig:qconcurrencyk10000} and \ref{fig:cconcurrencyk10000} for stack, queue and counter respectively. The reason for $k=10^4$ is to reduce the effect of contention due to small \width at lower $k$ values. This helps us focus on scalability effects. $2D$ algorithms' \width is configured as shown in Table \ref{table:dsrlx}. \random and \randomc algorithms' \width is set to $3P$, as the optimal balance between throughput and \accuracy \cite{don2dstackreport} since both of them do not provide a deterministic $k$ \relaxation bound but rather a probabilistic one where applicable \cite{Alistarh:2017:PCP:3087801.3087810,DBLP:conf/esa/Williams0D21}.

\kstack and \qsegment maintain a constant segment size as the number of \processors increases. This increases the rate at which segments get filled up, leading to a high frequency of \hopos and segment maintenance cost especially for asymmetric workloads. As observed, throughput gain quickly saturates even for a lower number of \processors leading to limited scalability.

The scalability of \lru is mostly limited by the global counter used to calculate the last recently used \subq. For every operation, the \processor has to increment the global counter using a \faa instruction, turning the counter into a scalability bottleneck. This can be observed when \lru performance is compared to that of a single \faa counter (\cfaa). \wfqueue suffers from the same \faa counter sequential bottleneck.

\tsstack's throughput is limited by the \popo search retries, searching for the newest item. Moreover, \popo operations might contend on the same newest items if there are not enough concurrent \pusho operations. Also, \popo lacks locality, which explains the drop in throughput on the \multisocket machine, due to the high inter-socket communication costs. We observe that throughput increases with increased \pusho rate. This is due to increased local processing and increased number of generated young items, leading to reduced \processor contention for \popo operations.

For \roundrobin algorithms, the \width is inversely proportional to the number of \processors (See Theorem~\ref{th:srobin}). As the number of \processors increases, \width reduces leading to increased contention. This explains the observed drop in throughput for a high number of \processors, especially for the \srobin and the \crobin algorithms due to their \substructure single \access design. The effect of lack of locality can be reduced by hardware pre-fetching, a feature available on both machines. This can also explain the \roundrobin better performance compared to the performance of the other algorithms that lack locality.

\random and \randomc algorithms are affected by the lack of locality, which is evident by the difference between \singlesocket and \multisocket results. We observe that the performance difference between \random and $2D$ algorithms increases on the \multisocket machine as compared to that on the \singlesocket machine. This demonstrates how much $2D$ algorithms gain from exploiting locality when executing on a \multisocket machine. Locality helps to avoid paying the high inter-socket communication cost through improved caching behaviour \cite{Hackenberg:2009:CCA:1669112.1669165,David:2013:EYA:2517349.2522714,Schweizer:2015:ECA:2923305.2923811,DBLP:conf/spaa/RukundoAT22}.


\section{Conclusion}
\label{sec:conclusion}
In this work, we have shown that semantics \relaxation has the potential to monotonically trade relaxed semantics of concurrent data structures for achieving throughput performance within tight \relaxation bounds. This has been achieved through an efficient two-dimensional framework that is simple and easy to implement for different data structures. 
We demonstrated that, by deriving two-dimensional lock-free designs for stacks, FIFO queues, dequeues and shared counters. 


Our experimental results have shown that \relaxing in one dimension, restricts the capability to control \relaxation behaviour in-terms of throughput and \accuracy. Compared to previous solutions, our framework can be used to extend existing data structures with minimal modifications while achieving better performance in terms of throughput and \accuracy. 

\bibliographystyle{plain}
\bibliography{mybib}
\end{document}